\documentclass[envcountsame]{llncs}

\usepackage{hyperref}
\usepackage{colortbl}
\usepackage{xspace}
\usepackage{amsmath}
\usepackage{amssymb}
\usepackage{ifthen}
\usepackage{xspace}
\usepackage[dvipsnames]{xcolor}
\usepackage{subcaption}
\usepackage{dcolumn}
\usepackage{tikz}
\usepackage{booktabs}
\usepackage{graphicx}
\usepackage[linesnumbered,ruled,noend,vlined]{algorithm2e}
\usepackage[capitalise]{cleveref}
\usepackage{environ}
\usepackage{wrapfig}
\usepackage{thm-restate}
\usepackage{multicol}
\usepackage{xcolor}
\usepackage{paralist}
\NewEnviron{killcontents}{}

\InputIfFileExists{cutmargins.tex}{%
	\pdfpagesattr{/CropBox [92 112 523 758]} 
}{%
}

\usepackage{newtxmath}
\usepackage{newtxtext}

\newif\ifHP
\HPtrue

\usetikzlibrary{automata,positioning,trees,shapes,petri,arrows,snakes,backgrounds,calc,fit,arrows.meta,shadows,math,decorations.markings}
\usetikzlibrary{decorations.pathmorphing}
\usetikzlibrary{shapes.multipart,arrows,automata}
\usetikzlibrary{shapes.geometric}
\usetikzlibrary{positioning}
\usetikzlibrary{calc,patterns,patterns.meta}

\def\blindreview{}

\newif\ifTR
\ifTR
\else
\fi

\newcounter{claimcounter}

\crefname{claimcounter}{Claim}{Claims}

\Crefname{algocf}{Algorithm}{Algorithms}

\usetikzlibrary{automata}
\usetikzlibrary{arrows.meta}
\usetikzlibrary{bending}
\usetikzlibrary{shapes.callouts}
\usetikzlibrary{quotes}
\usetikzlibrary{positioning}
\usetikzlibrary{calc}
\usetikzlibrary{matrix}

\definecolor{spotblue}{RGB}{31,120,180}
\definecolor{spotpink}{RGB}{255,77,160}
\definecolor{spotorange}{RGB}{255,127,0}
\definecolor{spotpurple}{RGB}{106,61,154}
\definecolor{spotgreen}{RGB}{51,160,44}
\definecolor{spotred}{RGB}{227,26,28}
\definecolor{spotyellowish}{RGB}{196,196,0}
\definecolor{spotgray}{RGB}{80,80,80}
\definecolor{spotlight blue}{RGB}{107,246,255}
\definecolor{spotlight pink}{RGB}{255,154,255}
\definecolor{spotlight orange}{RGB}{255,156,103}
\definecolor{spotlight purple}{RGB}{178,164,255}
\definecolor{spotlight green}{RGB}{167,237,121}
\definecolor{spotlight red}{RGB}{255,104,104}
\definecolor{spotlight yellowish}{RGB}{255,224,64}
\definecolor{spotlight gray}{RGB}{192,192,144}

\tikzset{
  >={Stealth[round,bend]},
}

\tikzstyle{automaton}=[
  semithick,shorten >=0pt,>={Stealth[round,bend]},
  node distance=1.5cm,
  initial text=,
  every initial by arrow/.style={every node/.style={inner sep=0pt}},
  every state/.style={
    align=center,
    fill=white,
    minimum size=7.5mm,
    inner sep=0pt,
    execute at begin node=\strut,
  },%
]
\tikzset{
  scc/.style={draw=gray,fill=black!10,rounded corners=2mm},
  lstate/.style={state},
  }

\tikzstyle{smallautomaton}=[
  automaton,
  node distance=7mm,
  every state/.style={minimum size=4mm,fill=white,inner sep=1.5pt}
]

\tikzstyle{mediumautomaton}=[
  automaton,
  node distance=1cm,
  every state/.style={minimum size=6mm,fill=white,inner sep=2pt}
]

\tikzstyle{cstate}=[state,capsule,text width=,inner xsep=-5pt]
\tikzstyle{dot}=[fill=black,circle,minimum size=4pt,inner sep=0]

\makeatletter
\tikzoption{initial angle}{\tikzaddafternodepathoption{\def\tikz@initial@angle{#1}}}
\makeatother

\tikzstyle{initial overlay}=[every initial by arrow/.append style={overlay}]

\tikzstyle{unreachable} = [densely dotted]

\tikzstyle{acclabel} = [
  fill=yellow!20,
  rounded corners=3pt,
  draw=black,
  inner ysep=0pt,
  inner xsep=3pt,
]
\tikzstyle{ltllabel} = [acclabel,fill=darkgreen!20]

\tikzstyle{namelabel} = [
  execute at begin node = {$($},%
  execute at end node = {$)$}
]

\tikzstyle{matrix of states} = [
matrix of nodes,
every node/.style={state},
column sep=.8cm,
row sep=.8cm,
execute at empty cell={
  \node[transparent,name=
    \tikzmatrixname-\the\pgfmatrixcurrentrow-\the\pgfmatrixcurrentcolumn]{};]},
]

\tikzstyle{accset}=[
  circle,inner sep=.5pt,
  draw=none, 
  solid,
  collacc0, text=white,
  thin,
  anchor=center,
  minimum size=3.3mm,
  text width={}, 
  font=\bfseries\sffamily\footnotesize
]
\tikzstyle{accsquare}=[accset,rectangle,inner sep=1.9pt,rounded corners=0pt]

\tikzset{
  collacc0/.style={fill=spotblue},
  collacc1/.style={fill=spotpink},
  collacc2/.style={fill=spotorange},
  collacc3/.style={fill=spotpurple},
  collacc4/.style={fill=spotgreen},
  collacc5/.style={fill=spotred},
  collacc6/.style={fill=spotyellowish,draw=black,text=black},
  collacc7/.style={fill=spotgray},
  collacc8/.style={fill=spotlight blue,draw=black,text=black},
  collacc9/.style={fill=spotlight pink},
  collacc10/.style={fill=spotlight orange},
  collacc11/.style={fill=spotlight purple},
  collacc12/.style={fill=spotlight green},
  collacc13/.style={fill=spotlight red},
  collacc14/.style={fill=spotlight yellowish},
  collacc15/.style={fill=spotlight gray},
}

\pgfkeyssetvalue{/sacc/where}{center}
\tikzset{
  sacc where/.code={
    \pgfkeyssetvalue{/sacc/where}{#1}
  }
}

\tikzset{
  l/.pic={\node[outer sep=2pt] {#1};},%
  acc/.pic={\node[accset,collacc#1]{#1};},%
  accsq/.pic={\node[accsquare,collacc#1]{#1};},%
  eacc/.pic={\node[accset,collacc#1]{\emptyacc};},%
  eaccsq/.pic={\node[accsquare,collacc#1]{\emptyacc};},%
  sacc/.style = {
    append after command=
      {pic at (\tikzlastnode.\pgfkeysvalueof{/sacc/where}) {acc=#1}}
  },
  saccsq/.style = {
    append after command=
      {pic at (\tikzlastnode.\pgfkeysvalueof{/sacc/where}) {accsq=#1}}
  },
  esacc/.style = {
    append after command=
      {pic at (\tikzlastnode.\pgfkeysvalueof{/sacc/where}) {eacc=#1}}
  },
  esaccsq/.style = {
    append after command=
      {pic at (\tikzlastnode.\pgfkeysvalueof{/sacc/where}) {eaccsq=#1}}
  },
  pics/cacc/.style 2 args={%
    code={\node[accset,collacc#1]{#2};}%
  },%
  pics/caccsq/.style 2 args={%
      code={\node[accsquare,collacc#1]{#2};}%
    },%
  csacc/.style 2 args = {%
      append after command=%
        {pic at (\tikzlastnode.\pgfkeysvalueof{/sacc/where}) {cacc={#1}{#2}}}%
  },%
  csaccsq/.style 2 args = {%
      append after command=%
        {pic at (\tikzlastnode.\pgfkeysvalueof{/sacc/where}) {caccsq={#1}{#2}}}%
  },%
}

\def\markbaseline{-.33em}
\def\tacc#1{\tikz[baseline=\markbaseline]\pic{acc=#1};\xspace}
\def\tcacc#1#2{\tikz[baseline=\markbaseline]\pic{cacc={#1}{#2}};\xspace}

\def\emptyacc{\phantom{0}}

\tikzstyle{removed} = [opacity=0, overlay]
\tikzstyle{SCC} = [
  rounded corners,
  draw=black!50, very thin,
  fill=black!7
]
\tikzstyle{trivial} = [dashed]
\tikzstyle{sccname} = [red,anchor=south east,outer xsep=13pt]

\newcommand{\vh}[1]{\textcolor{orange}{\ifmmode \text{[#1]}\else [VH: #1] \fi}}
\newcommand{\ol}[1]{\textcolor{blue}{\ifmmode \text{[OL: #1]}\else [OL: #1] \fi}}
\newcommand{\bs}[1]{\textcolor{teal}{\ifmmode \text{[BS: #1]}\else [BS: #1] \fi}}

\newcommand{\blinded}[1]{\ifx\blindreview\undefined #1 \else \textcolor{black!65}{[blinded for review]}\fi}

\newcommand{\reach}[0]{\mathit{reach}}
\newcommand{\reachof}[1]{\reach(#1)}
\newcommand{\reachofin}[2]{\reach_{#1}(#2)}
\newcommand{\rankof}[1]{\mathit{rank}(#1)}

\newcommand{\partto}[0]{\mathrel{\rightharpoonup}}

\newcommand{\myinf}[0]{\mathit{inf}}
\newcommand{\infof}[1]{\myinf(#1)}
\newcommand{\infofcol}[1]{\myinf_{\!\!\colourset}(#1)}
\newcommand{\Inf}[0]{\mathsf{Inf}}
\newcommand{\Infof}[1]{\Inf(#1)}
\newcommand{\Fin}[0]{\mathsf{Fin}}
\newcommand{\Finof}[1]{\Fin(#1)}

\newcommand{\inits}[0]{I}
\newcommand{\colouring}[0]{\mathsf{p}}

\newcommand{\img}[0]{\mathrm{img}}
\newcommand{\imgof}[1]{\img(#1)}

\newcommand{\aut}[0]{\mathcal{A}}

\newcommand{\M}[0]{\mathcal{M}}

\newcommand{\I}[0]{\mathcal{I}}

\newcommand{\cT}{\mathcal{T}}

\newcommand{\states}[0]{\mathcal{Q}}

\newcommand{\autex}[0]{\mathcal{A}_{\mathit{ex}}}

\newcommand{\width}{\mathit{width}}

\newcommand{\levelmodels}[0]{\mathsf{LM}}

\newcommand{\transover}[1]{\overset{#1}{\rightarrow}}
\newcommand{\ltr}[1]{\transover{#1}}

\newcommand{\word}[0]{w}
\newcommand{\wordof}[1]{\seqof{\word}{#1}}
\newcommand{\dagg}[0]{\mathcal{G}}
\newcommand{\daggk}[0]{\mathcal{K}}
\newcommand{\daggh}[0]{\mathcal{H}}
\newcommand{\dagof}[1]{\dagg_{#1}}
\newcommand{\dagw}[0]{\dagof{\word}}

\newcommand{\seqof}[2]{#1_{#2}}

\newcommand{\alginfela}[0]{\textsc{CInfTela}\xspace}

\newcommand{\lang}[0]{\mathcal{L}}
\newcommand{\langof}[1]{\lang(#1)}

\newcommand{\acc}{F}
\newcommand{\trans}{\delta}

\newcommand{\bigO}[0]{\mathcal{O}}
\newcommand{\bigOof}[1]{\bigO(#1)}

\definecolor{rowgray}{gray}{0.85}

\newcommand{\domof}[1]{\text{dom}(#1)}

\DeclareRobustCommand{\accinf}[0]{\textrm{Inf}}
\DeclareRobustCommand{\accfin}[0]{\textrm{Fin}}
\DeclareRobustCommand{\accinfof}[1]{\accinf(#1)}
\DeclareRobustCommand{\accfinof}[1]{\accfin(#1)}

\newcommand{\alphabet}[0]{\Sigma}

\newcommand{\lex}[0]{\textbf{lex-min}}
\newcommand{\lexof}[1]{\lex(#1)}

\newcommand{\proofbegin}[0]{\noindent\textit{Proof. }}
\newcommand{\proofend}[0]{}

\newcommand{\fsucc}[0]{\mathrel{\transconsist_{\trans}^a}}

\newcommand{\finsucca}[1]{\sqsubseteq_{#1}^a}
\newcommand{\evenceil}[1]{\lfloor\!\!\lfloor #1 \rfloor\!\!\rfloor}
\newcommand{\tight}[0]{\mathit{tight}}
\newcommand{\mutight}[0]{\mu\textrm{-}\tight}
\newcommand{\Smutightof}[1]{(#1, \mu)\textrm{-}\tight}
\newcommand{\Smutight}[0]{\Smutightof{S}}

\newcommand{\modelsof}[1]{\M_{\overline{#1}}}
\newcommand{\modelsminof}[1]{\M^{\min}_{\overline{#1}}}

\DeclareRobustCommand{\acccond}[0]{\mathsf{Acc}}

\DeclareRobustCommand{\acccondof}[1]{\tacc{#1}}

\newcommand{\succunion}[0]{\mathsf{Succ}}
\newcommand{\succact}[0]{\mathsf{SuccAct}}
\newcommand{\succtrack}[0]{\mathsf{SuccTrack}}
\newcommand{\breakempty}[0]{\mathsf{EmptyBreak}}
\newcommand{\mst}[0]{\mathsf{M}}

\newcommand{\instan}[0]{\mathbb{S}}
\newcommand{\instanof}[2]{\instan_{#1}^{#2}}
\newcommand{\instanDeltavarphi}[0]{\instanof{\Delta}{\varphi}}
\newcommand{\instanrun}[0]{R}
\newcommand{\instancompl}[0]{\mathsf{FinCompl}}
\newcommand{\finrun}[0]{\textsf{Fin}-run\xspace}

\newcommand{\colourset}[0]{\Gamma}

\newcommand{\emersonlei}[0]{\mathbb{EL}}
\newcommand{\emersonleiof}[1]{\emersonlei(#1)}

\newcommand{\buchi}[0]{B\"{u}chi\xspace}

\newcommand{\succactinf}[0]{\succact^{\mathsf{inf}}}
\newcommand{\succtrackinf}[0]{\succtrack^{\mathsf{inf}}}
\newcommand{\breakemptyinf}[0]{\breakempty^{\mathsf{inf}}}

\newcommand{\minf}[0]{\M^{\mathsf{inf}}}
\newcommand{\minftrack}[0]{\minf_{\mathsf{Track}}}
\newcommand{\minfact}[0]{\minf_{\mathsf{Act}}}
\newcommand{\instinf}[0]{\instan^{\mathsf{inf}}}

\newcommand{\mytrue}[0]{\mathit{tt}}
\newcommand{\myfalse}[0]{\mathit{ff}}

\newcommand{\succacttrue}[0]{\succact^{\mytrue}}
\newcommand{\succtracktrue}[0]{\succtrack^{\mytrue}}
\newcommand{\breakemptytrue}[0]{\breakempty^{\mytrue}}

\newcommand{\mintrue}[0]{\M^{\mytrue}}
\newcommand{\instantrue}[0]{\instan^{\mytrue}}

\newcommand{\binf}{\land\mathsf{inf}}
\newcommand{\succactbinf}[0]{\succact^{\binf}}
\newcommand{\succtrackbinf}[0]{\succtrack^{\binf}}
\newcommand{\breakemptybinf}[0]{\breakempty^{\binf}}
\newcommand{\mstbinf}[0]{\mst^{\binf}}
\newcommand{\minbinf}[0]{\M^{\binf}}
\newcommand{\instbinf}[0]{\instan^{\binf}}

\newcommand{\taccj}[0]{\tcacc{11}{$j$}}
\newcommand{\taccgof}[1]{\tcacc{7}{$#1$}}      
\newcommand{\taccBj}[0]{B_j}
\newcommand{\taccGj}[0]{G_j}
\newcommand{\taccGjl}[0]{G_{j,\ell}}

\newcommand{\lnref}[1]{Line~\ref{#1}}
\newcommand{\lnsref}[2]{Lines~\ref{#1}--\ref{#2}}

\newcommand{\gendagofof}[3]{#1^{#2}_{#3}}
\newcommand{\gendagof}[2]{\gendagofof{\dagg}{#1}{#2}}
\newcommand{\gendag}[0]{\gendagof{\Delta}{\word}}
\newcommand{\gendagk}[0]{\gendagofof{\daggk}{\Delta}{\word}}
\newcommand{\gendagh}[0]{\gendagofof{\daggh}{\Delta}{\word}}

\newcommand{\transconsist}[0]{%
  \begin{tikzpicture}[line width=0.6pt,transform shape,scale=0.7]
    \draw[->] (0,0) -- (1.5em,0);
    \draw (0.6em,0ex) circle (0.5ex);
  \end{tikzpicture}%
}

\title{Complementation of Emerson-Lei Automata
\\(Technical Report)
}

\author{Vojtěch Havlena \and
        Ondřej Lengál \and
        Barbora Šmahlíková}

\institute{Faculty of Information Technology, Brno University of Technology, Brno,\\ Czech Republic}

\authorrunning{V. Havlena, O. Lengál, B. Šmahlíková} 

\begin{document}

\maketitle

\begin{abstract}
We give new constructions for complementing subclasses of
Emerson-Lei automata using modifications of rank-based B\"{u}chi automata
complementation.
In particular, we propose a~specialized rank-based construction for a~Boolean
combination of Inf acceptance conditions, which heavily relies on a~novel way
of a~run DAG labelling enhancing the ranking functions with models of the
acceptance condition.
Moreover, we propose a technique for complementing generalized Rabin automata, 
which are structurally as concise as general Emerson-Lei automata (but can have a larger acceptance condition). The construction is modular in the sense that it combines a given 
complementation algorithm for a condition~$\varphi$ in a way that 
the resulting procedure handles conditions of the form~$\accfin \land \varphi$.
The proposed constructions give upper bounds that are exponentially better than
the state of the art for some of the classes.
\end{abstract}

\vspace{-9.0mm}
\section{Introduction}\label{sec:label}
\vspace{-3.0mm}

Complementation of $\omega$-automata is an important operation in formal
verification with various applications, for example in model checking wrt
expressive temporal logics such as QPTL~\cite{KestenP95} or
HyperLTL~\cite{ClarksonFKMRS14}; testing language inclusion of
$\omega$-automata, or in decision procedures of various
logics~\cite{buchi1962decision,HieronymiMOSSS24}.
For \buchi automata (BAs)---i.e.,
$\omega$-automata with the simplest acceptance condition---complementation has
been, from the theoretical point of view, thoroughly explored,
starting with constructions having the $2^{2^{\bigO(n)}}$ state
complexity~\cite{buchi1962decision} coming down to constructions asymptotically
matching the lower bound $\Omega((0.76n)^n)$ (modulo a~polynomial
factor)~\cite{Schewe09,AllredU18}.
Over the years, $\omega$-automata with more complex acceptance conditions (such
as generalized \buchi (GBAs), (generalized) Rabin/Streett, parity) have found
uses in practice.
The most general acceptance condition used is the so-called \emph{Emerson-Lei}
condition~\cite{EmersonL87}, which is
an~arbitrary Boolean formula consisting of $\Fin$ and $\Inf$ atoms.
$\Finof{\taccgof c}$ denotes that all transitions labeled with $\taccgof c$
must occur only finitely often in an accepting run and $\Infof{\taccgof c}$
denotes that there must be a~transition labeled with $\taccgof c$ occurring
infinitely often.
There are two main reasons for using more complex acceptance conditions:
\begin{inparaenum}[(i)]
  \item  more compact representation of automata and
  \item  the ability to determinize \mbox{(deterministic BAs are strictly less
    expressive than BAs).}
\end{inparaenum}

From the theoretical point of view, precise bounds on complementation of
automata with more complex acceptance condition is much less researched,
demonstrated by the best upper bound for (transition-based) Emerson-Lei
automata (TELAs) being $2^{2^{\bigO(n)}}$~\cite{SafraV89} states.
Here, the $\bigO$ in the exponent can hide a linear (or constant) factor, which
would have a doubly-exponential effect, giving little information about the
actual complexity.
In this paper, we present complementation algorithms for several subclasses of
TELAs and thoroughly study their complexity, giving better upper bounds than
the currently-best known algorithms.

Our contributions can be summarized as follows:
\begin{enumerate}
  \item  We propose a~rank-based complementation algorithm for $\accinf$-TELAs,
    i.e., TELAs where the acceptance condition does not contain any $\accfin$
    atom, with the state complexity $\bigOof{n (0.76nk)^n}$ where~$n$ is the
    number of states and $k$ is the number of \emph{minimal models} of the
    acceptance condition.

  \item  By instantiating the previously mentioned algorithm, we obtain
    a~complementation algorithm for generalized \buchi automata with~$k$
    colours constructing a~BA with the state complexity
    $\bigO(n(0.76nk)^n)$, which is, to the best of our knowledge, better than the best previously known
    algorithms.

  \item  We propose a~modular procedure for complementing TELAs with the
    acceptance condition $\accfinof{\taccgof c} \land \varphi$ given
    a~compatible complementation procedure for formula~$\varphi$.

  \item  Next, we instantiate the modular procedure to handle Rabin pairs
    ($\accfinof{\tacc 0} \land \accinfof{\tacc 1}$) and, in turn, obtain an
    algorithm for complementing Rabin automata with $k$ Rabin pairs with the
    complexity $\bigO(n^k(0.76n)^{nk})$, which is, again, better than any other
    algorithm that we know of.

  \item  Finally, we instantiate the procedure also for generalized
    Rabin pairs ($\accfinof{\tacc 0} \land \accinfof{\tacc 1} \land \ldots
    \land \accinfof{\tcacc 4 {$\ell$}}$) and obtain complementation
    constructions for generalized Rabin automata and TELAs with the
    upper bound $\bigO(n^{2^k}(0.76 nk)^{n2^k})$, which is the best upper bound
    for complementation of general TELAs that we are aware of.
\end{enumerate}

%

\vspace{-4.0mm}
\section{Preliminaries}
\vspace{-2.0mm}
We fix a~finite non-empty alphabet~$\Sigma$ and the first infinite
ordinal~$\omega$.
For~$k \in \omega$, we use~$\evenceil k$ to represent the largest even number less than or equal
to~$k$, e.g., $\evenceil{43} = \evenceil{42} = 42$.
An (infinite) word~$\word$ is a~function $\word\colon \omega \to \Sigma$ where
the $i$-th symbol is denoted as~$\wordof i$.
Sometimes, we represent~$\word$ as an~infinite sequence $\word = \wordof 0
\wordof 1 \dots$
We denote the set of all infinite words over~$\Sigma$ as $\Sigma^\omega$;
an \emph{$\omega$-language} is a~subset of~$\Sigma^\omega$.
We use $\cdot$ for ellipsis, e.g., if interested only in the second component
\mbox{of a~triple, we may write the triple as $(\cdot, x, \cdot)$.}

\vspace{-2.0mm}
\subsection{Emerson-Lei Acceptance Conditions}
\vspace{-1.0mm}

Given a~set $\colourset = \{0, \ldots, k -1\}$ of~$k$
\emph{colours} (often depicted as \tacc{0}, \tacc{1}, etc.), we define the
set of \emph{Emerson-Lei acceptance conditions} $\emersonleiof \colourset$ as
the set of formulae constructed according to the following grammar:
\vspace{-2mm}
\begin{align*}
  \alpha ::= \mytrue \mid \myfalse \mid
             \Infof c \mid \Finof c \mid
             (\alpha \land \alpha) \mid (\alpha \lor \alpha)\\[-7mm]
\end{align*}
for $c \in \colourset$.
The \emph{satisfaction} relation $\models$ for a~set of colours~$M \subseteq
\colourset$ and a~condition~$\alpha$ is defined inductively as follows (for $c
\in \colourset$):
\vspace{-2mm}
\begin{align*}
  M \models \mytrue, &&
  M \models \Finof c & \text{~ iff ~} c \notin M, &
  M \models \alpha_1 \lor \alpha_2 & \text{~ iff ~} M \models \alpha_1
  \text{ or } M \models \alpha_2,
  \\
  M \not\models \myfalse, &&
  M \models \Infof c & \text{~ iff ~} c \in M,  &
  M \models \alpha_1 \land \alpha_2 & \text{~ iff ~} M \models \alpha_1
  \text{ and } M \models \alpha_2.\\[-7mm]
\end{align*}
%
If $M \models \alpha$, we say that~$M$ is a~\emph{model} of~$\alpha$ 
We denote by $|\alpha|$ the number of atomic conditions contained in $\alpha$,
where multiple occurrences of the same atomic condition are counted multiple
times.

\pagebreak

\newcommand{\figInfAutExWrap}[0]{
\begin{wrapfigure}[10]{r}{25mm}
\vspace*{-6mm}
\hspace*{-3mm}
\begin{minipage}{25mm}
\begin{tikzpicture}[->,>=stealth',shorten >=0pt,auto,node distance=1.5cm,
                    scale=0.8,transform shape,initial text={}]
  \tikzstyle{every state}=[inner sep=3pt,minimum size=5pt]
  \tikzstyle{empty}=[]
  \tikzstyle{initstate}=[fill=yellow!30]

  \node[state,initial] (q) {$q$};
  \node[state,right of=q] (r) {$r$};
  \node[state,above of=r] (s) {$s$};
  \node[state,below of=r] (t) {$t$};

  \path (q) edge[loop above]  node {$a,b,c$} (q)
        (q) edge pic {acc=0} node[yshift=1mm] {$c$} (r)
        (r) edge node[right,xshift=1mm] {$a$}(s)
        (r) edge pic {acc=2} node[xshift=1mm] {$a$} (t)
        (s) edge[loop left]  node {$a$} (s)
        (t) edge node {$a,b$} (q)
        (s) edge pic {acc=1} node[above left,xshift=2mm,yshift=1mm] {$a$} (q);

\end{tikzpicture}

\centering
\scalebox{0.8}{
$\accinfof{\tacc 0} \land \accinfof{\tacc 1}$
}
\end{minipage}
\vspace{-2mm}
\caption{$\autex$}
\label{fig:ex_inf_inf_aut}
\end{wrapfigure}
}

\newcommand{\figInfRunDagExWrap}[0]{
\begin{wrapfigure}[19]{r}{39mm}
\vspace*{-6mm}
\hspace*{-2mm}
\begin{minipage}{38mm}
\begin{tikzpicture}[>=stealth',shorten >=0pt,auto,node distance=1.0cm,
                    scale=0.8,transform shape,initial text={}]
  \tikzstyle{every state}=[inner sep=3pt,minimum size=5pt,rectangle,rounded corners=1mm]
  \tikzstyle{empty}=[]

  \node[state]   (q0) {$q,0$};

  \node[state,below of=q0] (q1) {$q,1$};
  \coordinate[above left of=q1,xshift=3mm,yshift=-3mm] (q1_al);
  \coordinate[above right of=q1,xshift=-3mm,yshift=-3mm] (q1_ar);
  \node[state,right of=q1] (r1) {$r,1$};

  \node[state,below of=q1] (q2) {$q,2$};
  \coordinate[above right of=q2,xshift=-3mm,yshift=-3mm] (q2_ar);
  \node[below of=r1]       (r2) {};
  \node[state,right of=r2] (t2) {$t,2$};
  \coordinate[above right of=t2,xshift=-3mm,yshift=-3mm] (t2_ar);
  \node[state,right of=t2] (s2) {$s,2$};

  \node[below of=r2]       (r3) {};
  \node[state,below of=q2] (q3) {$q,3$};
  \node[state,below of=s2] (s3) {$s,3$};

  \node[state,below of=q3] (q4) {$q,4$};
  \node[state,right of=q4] (r4) {$r,4$};
  \node[right of=r4] (t4) {};
  \coordinate[above right of=t4,xshift=-3mm] (t4_ar);

  \node[state,below of=q4] (q5) {$q,5$};
  \node[below of=r4]       (r5) {};
  \node[state,right of=r5] (t5) {$t,5$};
  \node[state,right of=t5] (s5) {$s,5$};

  \node[below of=r5]      (r6) {};
  \node[state,left of=r6] (q6) {$q,6$};
  \node[below of=s5]      (s6) {};

  \node[below of=q6] (q7) {$\vdots$};
  \coordinate[below left of=q7,xshift=3mm,yshift=3mm] (q7_bl);
  \node[right of=q7] (r7) {$\ddots$};
  \node[right of=r7] (t7) {};
  \coordinate[below right of=t7,xshift=-3mm,yshift=3mm] (t7_br);
  \node[below of=s6]      (s7) {$\vdots$};



  \draw[->] (q0) edge pic[pos=0.4] {acc=0} (r1)
    (q0) edge (q1)
    (r1) edge (s2)
    (r1) edge 
      (t2)
    (q1) edge (q2);

  \draw[->] (s2) edge pic {acc=1} (q3)
    (s2) edge (s3)
    (t2) edge (q3)
    (q2) edge (q3);

  \draw[->] (q3) edge (q4)
    (q3) edge pic[pos=0.4] {acc=0} (r4)
    (q4) edge (q5)
    (r4) edge 
      (t5)
    (r4) edge (s5);

  \draw[->] (t5) edge (q6)
    (q5) edge (q6);

  \draw[->] (q6) edge (q7)
    (q6) edge pic[pos=0.4] {acc=0} (r7);

  \node[empty, node distance=8mm,below left of=q0,xshift=-3mm] (sym1) {$c$};
  \node[empty, below of=sym1] (sym2) {$a$};
  \node[empty, below of=sym2] (sym3) {$a$};
  \node[empty, below of=sym3] (sym4) {$c$};
  \node[empty, below of=sym4] (sym5) {$a$};
  \node[empty, below of=sym5] (sym6) {$b$};
  \node[empty, below of=sym6] (sym7) {$c$};
  \node[empty, below of=sym7,node distance=4.3mm] (symdots) {$\vdots$};

  \begin{pgfonlayer}{background}
    \node[draw,dashed,rectangle,fill=YellowGreen!60,draw=black!70,rounded corners=8pt,inner sep=3pt,fit=(s3) (s7)] (a) {};
    \fill[draw,dashed,rectangle,fill=YellowGreen!30!blue!20,draw=black!70,rounded corners=8pt,inner sep=3pt]
      (q1_al) -- (q7_bl) -- (t7_br) -- (t2_ar) -- (q2_ar) -- (q1_ar) --cycle;
  \end{pgfonlayer}

  \node[empty, below of=s5, node distance=8mm, text=black] (rank0) {\textit{rank 0}};
  \node[empty, right of=q6, node distance=15mm, text=black] (rank1) {$\begin{array}{c}\textit{rank 1}\\\textit{model:}\,\{\tacc 1\} \end{array}$};
  \node[empty, right of=q0, node distance=18mm, yshift=-4mm, text=black] (rank2) {\textit{rank 2}};

\end{tikzpicture}
\end{minipage}
\vspace{-2mm}
\caption{A labelled run DAG of~$\autex$ over the word
  $caa(cab)^\omega \notin \langof \autex$}
\label{fig:ex_inf_inf_rundag}
\end{wrapfigure}
}

\vspace{-2.0mm}
\subsection{Emerson-Lei Automata}
\vspace{-1.0mm}
%

\begin{wrapfigure}[10]{r}{25mm}
\vspace*{-6mm}
\hspace*{-3mm}
\begin{minipage}{25mm}
\begin{tikzpicture}[->,>=stealth',shorten >=0pt,auto,node distance=1.5cm,
                    scale=0.8,transform shape,initial text={}]
  \tikzstyle{every state}=[inner sep=3pt,minimum size=5pt]
  \tikzstyle{empty}=[]
  \tikzstyle{initstate}=[fill=yellow!30]

  \node[state,initial] (q) {$q$};
  \node[state,right of=q] (r) {$r$};
  \node[state,above of=r] (s) {$s$};
  \node[state,below of=r] (t) {$t$};

  \path (q) edge[loop above]  node {$a,b,c$} (q)
        (q) edge pic {acc=0} node[yshift=1mm] {$c$} (r)
        (r) edge node[right,xshift=1mm] {$a$}(s)
        (r) edge pic {acc=2} node[xshift=1mm] {$a$} (t)
        (s) edge[loop left]  node {$a$} (s)
        (t) edge node {$a,b$} (q)
        (s) edge pic {acc=1} node[above left,xshift=2mm,yshift=1mm] {$a$} (q);

\end{tikzpicture}

\centering
\scalebox{0.8}{
$\accinfof{\tacc 0} \land \accinfof{\tacc 1}$
}
\end{minipage}
\vspace{-2mm}
\caption{$\autex$}
\label{fig:ex_inf_inf_aut}
\end{wrapfigure}

A~(nondeterministic) transition-based\footnote{%
We only consider transition-based acceptance in order to avoid cluttering the
paper by dealing with accepting states \emph{and} accepting transitions.
Extending our approach to state/transition-based (or just state-based) automata
is straightforward.
}
\emph{Emerson-Lei automaton} (TELA)
over~$\Sigma$ is a~tuple $\aut = (\states, \trans, \inits, \colourset,
\colouring, \acccond)$,
where $\states$ is a~finite set of \emph{states} (we often use~$n$ to denote $|\states|$),
$\trans \subseteq \states \times \Sigma \times \states$ is a~set of
\emph{transitions}\footnote{%
Note that there is also a~more general definition
of TELAs with $\trans \subseteq \states \times \Sigma \times 2^{\colourset}
\times \states$; in this paper, we use the simpler one.},
$\inits \subseteq \states$ is the set of \emph{initial} states,
$\colourset$ is the set of \emph{colours},
$\colouring\colon \trans \to 2^{\colourset}$ is a~\emph{colouring} of transitions, and
$\acccond \in \emersonleiof \colourset$.
We use $p \ltr a q$ to denote that $(p,a,q) \in \trans$ and sometimes
treat~$\trans$ as a~function $\trans\colon \states \times
\alphabet \to 2^{\states}$.
Moreover, we extend~$\trans$ to sets of states $P \subseteq \states$ as
$\trans(P, a) = \bigcup_{p \in P} \trans(p,a)$.
See \cref{fig:ex_inf_inf_aut} for an example TELA~$\autex$ over~$\Sigma =
\{a,b,c\}$ with 3 colours $\colourset = \{\tacc 0, \tacc 1, \tacc 2\}$ and the
acceptance condition $\accinfof{\tacc 0} \land \accinfof{\tacc 1}$.
We define $|\aut| = |\states|$.

A~\emph{run}
of~$\aut$ from~$q \in \states$ on an input word~$\word$ is an infinite sequence $\rho\colon
\omega \to \states$ that starts in~$q$ and respects~$\trans$, i.e., $\rho(0) = q$ and
$\forall i \geq 0\colon \rho(i) \ltr{\wordof i}\rho(i+1) \in \trans$.
Let $\infof \rho \subseteq \delta$ denote the set of
transitions occurring in~$\rho$ infinitely often and $\infofcol \rho =
\bigcup\{\colouring(x) \mid x \in \infof \rho\}$ be the set of infinitely
often occurring colours.
A~run~$\rho$ is \emph{accepting} wrt an acceptance condition~$\alpha$, written
as $\rho \models \alpha$, iff $\infofcol \rho \models \alpha$ and
$\rho$ is accepting in~$\aut$ iff $\rho \models \acccond$.
The \emph{language} of~$\aut$, denoted as~$\langof{\aut}$, is defined as the set of words
$w \in \alphabet^\omega$ for which there exists an accepting run in~$\aut$
starting with some state in~$\inits$.
Classical acceptance conditions can be in this more general framework described
as follows (we only provide those used later in the paper and include their abbreviations):
\vspace{-3mm}
\begin{itemize}
  \setlength{\itemsep}{0.5mm}
	\item \emph{\buchi} (BA): $\acccond = \accinfof{\tacc 0}$,
	\item \emph{co-\buchi} (CBA): $\acccond = \accfinof{\tacc 0}$,
	\item \emph{Generalized \buchi} (GBA): $\acccond = \bigwedge_{0 \leq j < k} \accinfof{\taccj}$,
	\item \emph{Generalized co-\buchi} (GCBA): $\acccond = \bigvee_{0 \leq j < k} \accfinof{\taccj}$,
  \item \emph{Rabin}: $\bigvee_{0 \leq j < k} \accfinof{\taccBj} \land \accinfof{\taccGj}$, and
  \item \emph{Generalized Rabin}: $\bigvee_{0 \leq j < k} (\accfinof{\taccBj} \land \bigwedge_{0 \leq \ell < m_j}\accinfof{\taccGjl})$.
  \item \emph{Parity}\footnote{
    We consider the so-called \emph{parity min odd} condition; any parity condition from the set $\{\text{min}, \text{max}\} \times \{\text{even}, \text{odd}\}$ can be easily translated to it.
    }: $\accfinof{\tacc 0} \land (\accinfof{\tacc 1} \lor (\accfinof{\tacc 2} \land (\accinfof{\tacc 3} \lor (\accfinof{\tacc 4} \land \ldots))))$
\end{itemize}
\vspace{-3mm}
Furthermore, we use $\accinf$-TELA to denote a~TELA where the acceptance
condition contains no $\accfin$ atoms.
We also use the syntactic sugar~$\aut = (\states, \trans, \inits, \acc)$ to
denote a~(transition-based) BA that would be defined using the TELA definition
above as $(\states, \trans, \inits, \{\tacc 0\}, \{t \mapsto \emptyset \mid t
\in \trans \setminus \acc\} \cup \{t \mapsto \{\tacc 0\} \mid t \in \acc\}, \Infof{\tacc 0})$.

%

\newcommand{\algrundaginf}[0]{
\begin{algorithm}[t]
	\caption{$\accinf$-TELA run DAG labelling}
	\label{alg:label-run-dag}

	\SetKwInput{KwInput}{Input}
	\SetKwInput{KwOutput}{Output}

	\KwInput{A~run DAG $\dagw$ of~$\aut$ over $\word$, acceptance condition~$\acccond$}
	\KwOutput{A~run DAG ranking $r$ and a~model assignment $m$ if $\word \notin \langof \aut$, else $\bot$}

	$i \gets 0$, $r \gets \emptyset$, $m \gets \emptyset$\tcp*[r]{$i \in \omega$, $r\colon V \partto \{0, \ldots, 2|\states|\}$, $m\colon V \partto \modelsminof \acccond$ }
  $\dagg^0 = (V^0, E^0) \gets \dagw$ without finite
  vertices\;\label{ln:remove_finite_init}
  \ForEach{$v \in \dagw$ s.t.\ $v$ is finite}{
    $r(v) \gets 0$, $m(v) \gets \lexof{\overline{\acccond}}$\;\label{ln:initial_finite}
  }
	\While{$\dagg^i \neq \emptyset$} {\label{ln:main_loop}
    \If{$\exists (\eta\colon U \to \modelsminof \acccond)$ s.t. $U \subseteq V^i$ and $\eta$ is endangered in $\dagg^i$}{\label{ln:endangered}
    \ForEach{$v \in U$ and $u \in \reachofin{\dagg^i}{v}$}{\label{ln:step_one}
      $r(u) \gets i + 1$, $m(u) \gets \eta(v)$\;
    }
		$\dagg^{i+1} \gets \dagg^{i}$ without vertices with the rank $i+1$\;\label{ln:remove_endangered}
    \ForEach{$v \in \dagg^{i+1}$ s.t.\ $v$ is finite in $\dagg^{i+1}$}{
    \label{ln:step_two}
      $r(v) \gets i + 2$, $m(v) \gets \lexof{\overline{\acccond}}$\;
    }
    $\dagg^{i+2} \gets \dagg^{i+1}$ without vertices with the rank $i+2$\;\label{ln:remove_finite}
    $i \gets i + 2$\;\label{ln:increment}
    }
    \Else{
      \Return $\bot$\;
    }
	}
  \Return $(r,m)$\;
\end{algorithm}
}

\vspace{-3.0mm}
\subsection{Run DAGs}\label{sec:label}
\vspace{-2.0mm}

In this section, we recall the terminology from~\cite{HavlenaLS22} (which is
a~minor modification of the terminology from~\cite{KupfermanV01} and
\cite{Schewe09}) used heavily in the paper.
Let $\aut = (\states, \trans, \inits, \colourset, \colouring, \acccond)$ be
a~TELA.
We fix the definition of the \emph{run DAG} of~$\aut$ over a~word~$\word$ to be
a~DAG (directed acyclic graph) $\dagw = (V,E)$ of vertices~$V$ and edges~$E$
where
\vspace{-0mm}
\begin{itemize}
  \setlength{\itemsep}{0mm}
  \item  $V \subseteq \states \times \omega$ s.t. $(q, i) \in V$ iff there is
    a~run~$\rho$ of $\aut$ from $\inits$ over $\word$ with $\rho_i = q$,
  \item  $E \subseteq V \times V$ s.t.~$((q, i), (q',i')) \in E$ iff $i' = i+1$
    and $q' \in \delta(q, \wordof i)$.
\end{itemize}
\vspace{-0mm}


\begin{wrapfigure}[19]{r}{39mm}
\vspace*{-6mm}
\hspace*{-2mm}
\begin{minipage}{38mm}
\begin{tikzpicture}[>=stealth',shorten >=0pt,auto,node distance=1.0cm,
                    scale=0.8,transform shape,initial text={}]
  \tikzstyle{every state}=[inner sep=3pt,minimum size=5pt,rectangle,rounded corners=1mm]
  \tikzstyle{empty}=[]

  \node[state]   (q0) {$q,0$};

  \node[state,below of=q0] (q1) {$q,1$};
  \coordinate[above left of=q1,xshift=3mm,yshift=-3mm] (q1_al);
  \coordinate[above right of=q1,xshift=-3mm,yshift=-3mm] (q1_ar);
  \node[state,right of=q1] (r1) {$r,1$};

  \node[state,below of=q1] (q2) {$q,2$};
  \coordinate[above right of=q2,xshift=-3mm,yshift=-3mm] (q2_ar);
  \node[below of=r1]       (r2) {};
  \node[state,right of=r2] (t2) {$t,2$};
  \coordinate[above right of=t2,xshift=-3mm,yshift=-3mm] (t2_ar);
  \node[state,right of=t2] (s2) {$s,2$};

  \node[below of=r2]       (r3) {};
  \node[state,below of=q2] (q3) {$q,3$};
  \node[state,below of=s2] (s3) {$s,3$};

  \node[state,below of=q3] (q4) {$q,4$};
  \node[state,right of=q4] (r4) {$r,4$};
  \node[right of=r4] (t4) {};
  \coordinate[above right of=t4,xshift=-3mm] (t4_ar);

  \node[state,below of=q4] (q5) {$q,5$};
  \node[below of=r4]       (r5) {};
  \node[state,right of=r5] (t5) {$t,5$};
  \node[state,right of=t5] (s5) {$s,5$};

  \node[below of=r5]      (r6) {};
  \node[state,left of=r6] (q6) {$q,6$};
  \node[below of=s5]      (s6) {};

  \node[below of=q6] (q7) {$\vdots$};
  \coordinate[below left of=q7,xshift=3mm,yshift=3mm] (q7_bl);
  \node[right of=q7] (r7) {$\ddots$};
  \node[right of=r7] (t7) {};
  \coordinate[below right of=t7,xshift=-3mm,yshift=3mm] (t7_br);
  \node[below of=s6]      (s7) {$\vdots$};



  \draw[->] (q0) edge pic[pos=0.4] {acc=0} (r1)
    (q0) edge (q1)
    (r1) edge (s2)
    (r1) edge 
      (t2)
    (q1) edge (q2);

  \draw[->] (s2) edge pic {acc=1} (q3)
    (s2) edge (s3)
    (t2) edge (q3)
    (q2) edge (q3);

  \draw[->] (q3) edge (q4)
    (q3) edge pic[pos=0.4] {acc=0} (r4)
    (q4) edge (q5)
    (r4) edge 
      (t5)
    (r4) edge (s5);

  \draw[->] (t5) edge (q6)
    (q5) edge (q6);

  \draw[->] (q6) edge (q7)
    (q6) edge pic[pos=0.4] {acc=0} (r7);

  \node[empty, node distance=8mm,below left of=q0,xshift=-3mm] (sym1) {$c$};
  \node[empty, below of=sym1] (sym2) {$a$};
  \node[empty, below of=sym2] (sym3) {$a$};
  \node[empty, below of=sym3] (sym4) {$c$};
  \node[empty, below of=sym4] (sym5) {$a$};
  \node[empty, below of=sym5] (sym6) {$b$};
  \node[empty, below of=sym6] (sym7) {$c$};
  \node[empty, below of=sym7,node distance=4.3mm] (symdots) {$\vdots$};

  \begin{pgfonlayer}{background}
    \node[draw,dashed,rectangle,fill=YellowGreen!60,draw=black!70,rounded corners=8pt,inner sep=3pt,fit=(s3) (s7)] (a) {};
    \fill[draw,dashed,rectangle,fill=YellowGreen!30!blue!20,draw=black!70,rounded corners=8pt,inner sep=3pt]
      (q1_al) -- (q7_bl) -- (t7_br) -- (t2_ar) -- (q2_ar) -- (q1_ar) --cycle;
  \end{pgfonlayer}

  \node[empty, below of=s5, node distance=8mm, text=black] (rank0) {\textit{rank 0}};
  \node[empty, right of=q6, node distance=15mm, text=black] (rank1) {$\begin{array}{c}\textit{rank 1}\\\textit{model:}\,\{\tacc 1\} \end{array}$};
  \node[empty, right of=q0, node distance=18mm, yshift=-4mm, text=black] (rank2) {\textit{rank 2}};

\end{tikzpicture}
\end{minipage}
\vspace{-2mm}
\caption{A labelled run DAG of~$\autex$ over the word
  $caa(cab)^\omega \notin \langof \autex$}
\label{fig:ex_inf_inf_rundag}
\end{wrapfigure}

\noindent
See \cref{fig:ex_inf_inf_rundag} for an example of a~run DAG of~$\autex$ from
\cref{fig:ex_inf_inf_aut} over the word~$caa(cab)^\omega \notin \langof \autex$
(we will return to the additional labels in the figure later).
Given a~DAG~$\dagg = (V,E)$,
we often identify~$\dagg$ with~$V$, for instance, we will write $(p, i) \in
\dagg$ to denote that $(p, i) \in V$.
For a~vertex $v \in \dagg$, we denote the set of vertices of~$\dagg$ reachable
from~$v$ (including~$v$ itself) as $\reachofin{\dagg}{v}$ or just $\reachof v$
if~$\dagg$ is clear from the context.
A~vertex~$v \in \dagg$ is \emph{finite} iff $\reachof{v}$ is finite and
\emph{infinite} if it is not finite.
In \cref{fig:ex_inf_inf_rundag}, the vertices $(s,2), (s,3), (s,5), \ldots$ are
finite and all other vertices are infinite.
Moreover, for a~colour~$\taccgof{c} \in \colourset$, an edge $((q, i), (q',
i+1)) \in E$ is a~$\taccgof{c}$-edge if $\taccgof{c} \in \colouring(q
\ltr{\wordof i} q')$ and a~vertex~$v \in V$ is
$\taccgof{c}$-\emph{endangered} iff it cannot reach any $\taccgof{c}$-edge.
For a~set of colours $C \subseteq \colourset$, $v$~is
$C$-\emph{endangered} iff it is $\taccgof{c}$-endangered for every~$\taccgof{c} \in C$.
For example, in \cref{fig:ex_inf_inf_rundag}, the vertices $(q,1)$ and $(t,2)$
are $\{\tacc 1\}$-endangered but they are not $\{\tacc 0, \tacc
1\}$-endangered.
A~pair of vertices $v_1, v_2 \in V$ is \emph{converging} iff $\reachof{v_1}
\cap \reachof{v_2} \neq \emptyset$ ($v_1$~and~$v_2$ \emph{converge}).
A~function $r\colon V \to \omega$ is a \emph{run DAG ranking} if for every $v
\in V$ it holds that $\forall u \in \reachof v \colon r(u) \leq r(v)$.
We use $\max(r)$ to denote the \emph{rank} of~$r$, i.e., the maximum value from
$\{r(u) \mid u \in V\}$ if it exists and $\infty$ otherwise.
A ranking~$r$ of~$\dagg$ is called \emph{tight} iff
there exists a~level~$i$ such that
\begin{inparaenum}[(i)]
  \item $m = \max\{r((q,i)) \mid q \in \states\}$ is odd and 
  \item for all levels~$j \geq i$ it holds that $\{1,3,\dots,m\} \subseteq \{ r((q,j)) \mid q\in\states \}$. 
\end{inparaenum}
%


%

\vspace{-0.0mm}
\section{Complementation of $\accinf$-TELAs}\label{sec:inftela}
\vspace{-0.0mm}

In this section, we describe a~complement construction for $\accinf$-TELAs.
Our approach is an extension of rank-based BA
complementation
algorithms~\cite{KupfermanV01,FriedgutKV06,Schewe09,KarmarkarC09,ChenHL19,HavlenaL21,HavlenaLS22},
which construct a~BA whose runs simulate a~run DAG ranking procedure.
We start with giving the run DAG ranking procedure (which extends the ranking
procedure from~\cite{KupfermanV01} with the introduction of models) and then
proceed to the complementation algorithm itself. 
One can see our algorithm also as an improvement of the algorithm for
complementing GBAs in~\cite{KupfermanV05} by
\begin{inparaenum}[(i)]
  \item  introducing model assignments,
  \item  getting better complexity through the use of tight rankings, and
  \item  generalizing the construction from GBAs to arbitrary $\accinf$-TELAs.
\end{inparaenum}

\vspace{-0.0mm}
\subsection{$\accinf$-TELA Run DAG Labelling}
\label{sec:infelalabelling}
\vspace{-0.0mm}

Let $\aut = (\states, \trans, \inits, \colourset,
\colouring, \acccond)$ be an $\accinf$-TELA.
We use $\overline \acccond$ to denote the propositional formula obtained
from~$\acccond$ by replacing conjunctions by disjunctions and vice versa, and
substituting atoms of the form $\accinfof{\taccj}$ by~$\taccj$ (this can be viewed as
negating~$\acccond$, transforming it into the negation normal form, substituting
$\neg \accinfof{\taccj}$ by $\accfinof{\taccj}$, and denoting each $\accfinof{\taccj}$
just by~$\taccj$).
Let $\modelsof \acccond$ be the set of models of $\overline\acccond$ where the colours $\taccj$ are
interpreted as propositional variables.
For example, if $\acccond = \accinfof{\tacc 0} \land (\accinfof{\tacc 1} \lor \accinfof{\tacc 2})$, then $\overline
\acccond = \tacc 0 \lor (\tacc 1 \land \tacc 2)$ and $\modelsof \acccond = \{ \{ \tacc 0 \}, \{ \tacc 1, \tacc 2 \}, \{ \tacc 0, \tacc 1 \},\{\tacc 0, \tacc 2\}, \{\tacc 0, \tacc 1, \tacc 2\} \}$
($\modelsof \acccond$~can be interpreted as saying which combinations of $\accinf$-conditions
need to be broken in order to break~$\acccond$; in the example above, we can,
e.g., break $\accinfof{\tacc 0}$, we can break both $\accinfof{\tacc 1}$ and
$\accinfof{\tacc 2}$, etc.).
%
Furthermore, we use $\modelsminof \acccond$ to denote the set of \emph{minimal
models} of~$\overline \acccond$, i.e., $\modelsminof \acccond$ is the set where
\begin{inparaenum}[(i)]
  \item for every model $m \in \modelsof \acccond$, there exists a~model~$m' \in \modelsminof \acccond$ such that $m' \subseteq m$, and 
  \item there are no $m, m' \in \modelsminof \acccond$ such that $m \subset m'$ .
\end{inparaenum}
We note that~$\modelsof \acccond$ can be obtained as the upward closure
of~$\modelsminof \acccond$ (and $\modelsminof \acccond$ is an \emph{antichain}).
For the example acceptance condition above, $\modelsminof \acccond = \{\{\tacc
0\}, \{\tacc 1, \tacc 2\}\}$.
Moreover, we use $\lexof{\overline{\acccond}}$ to denote the lexicographically
smallest model from~$\modelsminof \acccond$ (w.l.o.g., we assume
$\modelsof \acccond \neq \emptyset$).
$\lexof{\overline{\acccond}}$ is used to pinpoint one model when any model can
be used.

Let $\dagg = (V,E)$ be a run DAG of~$\aut$ over~$\word$.
For a set of vertices
$U\subseteq V$, a~mapping $\eta\colon U\to \modelsminof \acccond$ is called
\emph{endangered in~$\dagg$} if
\begin{enumerate}
  \item  $\eta$ is finite and nonempty,
  \item  each $v\in U$ is $\eta(v)$-endangered in~$\dagg$, and
    \label{def:endangered-map:endangered}
  \item  for each pair of vertices $v_1, v_2\in U$ converging in~$\dagg$, we
    have $\eta(v_1) = \eta(v_2)$.
\end{enumerate}
%
%
A~function $m$ with the signature $m\colon V \to \modelsminof \acccond$ is called a~\emph{model
assignment}.
For instance, for $\autex$ in \cref{fig:ex_inf_inf_aut}, we have $\modelsminof
\acccond = \{\{\tacc 0\}, \{\tacc 1\}\}$ since $\autex$ is a~GBA.
In addition, for the run DAG in \cref{fig:ex_inf_inf_rundag} and a~set
$\{(q,1),(t,2)\}$, the mapping $\{(q,1) \mapsto \{\tacc 1\}, (t,2) \mapsto
\{\tacc 1\}\}$ is endangered in~$\dagg$.
On the other hand, there exists no endangered mapping for the set $\{(s,2)\}$ in~$\dagg$, as $(s,2)$ can reach both a~$\tacc 0$-edge as well as a~$\tacc 1$-edge.


\algrundaginf

In \cref{alg:label-run-dag}, we give a~(nondeterministic) ranking procedure that
assigns ranks and minimal models of~$\overline\acccond$ to each
vertex of~$\dagg$.
Intuitively, the algorithm starts by giving all initially finite vertices the
rank~0 and assigning their model to~$\lexof{\overline{\acccond}}$ (\lnref{ln:initial_finite}).
Then, it proceeds in iterations, each starting with the DAG~$\dagg^i$ and
consisting of two steps:
\begin{enumerate}
  \item  First, the algorithm tries to find a~model assignment $\eta\colon U
    \to \modelsminof \acccond$ for a~finite nonempty set of vertices~$U$
    of~$\dagg^i$ s.t.\ for all~$u \in U$, if $\eta(u) = \{\taccgof{c_1},
    \ldots, \taccgof{c_\ell}\}$, then every
    path starting in~$u$ satisfies the condition $\bigwedge_{1 \leq j \leq
    \ell}\accfinof{\taccgof{c_j}}$ (the path breaks all the 
    $\accinfof{\taccgof{c_j}}$ conditions, i.e., $\eta$~is endangered).
    If such a~model assignment exists, the algorithm assigns rank~$i+1$ to all
    vertices reachable from~$U$ and removes them from the DAG, creating
    DAG~$\dagg^{i+1}$ (\lnsref{ln:step_one}{ln:remove_endangered}).

  \item  Second, the algorithm assigns rank~$i+2$ to all vertices
    in~$\dagg^{i+1}$ that became finite (after the previous step) and removes
    them from the DAG, creating DAG~$\dagg^{i+2}$ (\lnsref{ln:step_two}{ln:remove_finite}).
    The counter~$i$ is incremented by two and the next iteration continues.

\end{enumerate}
The iterations end when~$\dagg^i$ is empty or when no suitable model
assignment~$\eta$ is found (which happens when~$w$ is accepted by~$\aut$).
Note that due to the nondeterminism within the algorithm, it may be possible to
obtain, in two different runs of the algorithm on the same run DAG, two
different pairs $(r_1, m_1)$ and $(r_2, m_2)$ with $\max(r_1) \neq \max(r_2)$.



\begin{example}
See \cref{fig:ex_inf_inf_rundag} for a~possible labelling of the run DAG
  of~$\autex$ over the word $caa(cab)^\omega$.
  The ranking procedure proceeds in the following steps:
  \begin{enumerate}
    \item  $(i=0)$ First, all finite vertices, which are in this example
      vertices of the form $(s,3)$, $(s,5)$, \ldots, $(s, 3j+2)$ for all $1
      \leq j$, are assigned rank~0 and model $\lexof{\overline{\acccond}}$, and $\dagg^0$ is set
      to be $\dagg^w$ without those vertices.
      (\lnsref{ln:remove_finite_init}{ln:initial_finite})

    \item  Second, we set $\eta_1$ to the mapping
      $\eta_1 = \{(q,1) \mapsto \{\tacc 1\}, (t,2) \mapsto \{\tacc 1\}\}$.
      The mapping $\eta_1$ is endangered in~$\dagg^0$ because the following conditions hold:
      \begin{enumerate}
        \item  $\eta_1$ is finite and nonempty,
        \item  neither $(q,1)$ nor $(t,2)$ can reach a~$\tacc 1$ transition,
          and
        \item  $(q,1)$ and $(t,2)$ converge (in $(q,3)$) and they are both
          assigned the same model ($\eta_1((q,1)) = \eta_1((t,2)) = \{\tacc 1\}$).
      \end{enumerate}
      In particular, $\eta_1$ is the endangered mapping that gives the largest
      number of vertices of~$\dagg^0$ rank~1.
      (\lnref{ln:endangered})
    \item  Third, we assign every vertex in~$\dagg^0$ reachable from~$(q,1)$
      or~$(t,2)$ the rank~1 and model~$\{\tacc 1\}$.
      (\lnref{ln:step_one})
    \item  Fourth, we obtain~$\dagg^1$ from~$\dagg^0$ by removing vertices with
      rank~1. (\lnref{ln:remove_endangered})
    \item  $\dagg^1$ contains three vertices ($\{(q,0), (r,1), (s,2)\}$), which all get rank~$2$
      (\lnref{ln:step_two}) and are removed in~$\dagg^2$
      (\lnref{ln:remove_finite}). The ranking procedure finishes.\qed
  \end{enumerate}
\end{example}

%


\begin{lemma}\label{lem:bot-no-accept}
If \cref{alg:label-run-dag} returns~$\bot$, then $\word \in \langof{\aut}$.
\end{lemma}

\proofbegin
  Let $\acccond'$ be a~formula in the disjunctive normal form (DNF) equivalent
  to~$\acccond$, i.e., $\acccond' = \bigvee_{j=1}^\ell \varphi_j$ where $\varphi_j =
  \accinfof{\taccgof{c^j_1}} \land \dots \land \accinfof{\taccgof{c^j_{k_j}}}$
  for some~$\ell$ and $k_1, \ldots, k_\ell$.
  Note that $\modelsminof \acccond = \modelsminof {\acccond'}$ contains sets of colours $M \subseteq
  \colourset$, each of them with at
  least one colour from~$\varphi_1$, at least one colour from~$\varphi_2$, etc.
  In order for \cref{alg:label-run-dag} to return~$\bot$, it needs to hold that
  there is some~$i \geq 0$ such that $\dagg^i$ is nonempty and there does not
  exist any mapping $\eta\colon U \to \modelsminof \acccond$, with $U \subseteq V^i$, that
  would be endangered in $\dagg^i$.
  In~particular, such a~(nonempty) mapping~$\eta$ does not exist iff no vertex
  $v \in \dagg^i$ satisfies point~(\ref{def:endangered-map:endangered}) of the
  definition of an endangered mapping (i.e., when we can find an accepting path
  from all vertices remaining in~$\dagg^i$).
  Therefore, it follows that no vertex $v \in \dagg^i$ is $M$-endangered for any $M
  \in \modelsminof \acccond$, i.e., in other words,

  \vspace{3mm}
  \begin{minipage}{0.8\textwidth} \label{eq:every-vertex}
    for every vertex $v \in \dagg^i$ there is some clause~$\varphi_j$ such that~$v$
    can in~$\dagg^i$ reach a~$\taccgof{c^j_p}$-edge for each $1 \leq p \leq k_j$.
  \end{minipage}
  \hfill
  (Reach)
  \vspace{3mm}

  \begin{wrapfigure}[13]{r}{2cm}
  \vspace*{-9mm}
  \hspace*{0mm}
  \begin{minipage}{2cm}
  \begin{tikzpicture}[>=stealth',shorten >=0pt,auto,node distance=1.0cm,
                    scale=0.8,transform shape,initial text={}]
  \tikzstyle{every state}=[inner sep=3pt,minimum size=5pt,rectangle,rounded corners=1mm]
  \tikzstyle{mystate}=[inner sep=3pt,minimum size=5pt,circle,draw]
  \tikzstyle{empty}=[]

  \node[state] (q0) {$q,0$};

  \node[state,below left of=q0,yshift=-5mm] (r1) {$r,1$};
  \node[state,below right of=q0,yshift=-5mm] (s1) {$s,1$};
  \node[state,below right of=r1,yshift=-5mm] (q2) {$q,2$};
  \node[state,below left of=q2,yshift=-5mm] (r3) {$r,3$};
  \node[state,below right of=q2,yshift=-5mm] (s3) {$s,3$};
  \node[below right of=r3,yshift=-5mm] (q4) {$\vdots$};

  \draw[->] (q0) edge pic {acc=0} (r1)
            (q0) edge pic {acc=1} (s1)
            (r1) edge (q2)
            (s1) edge (q2)
            (q2) edge pic {acc=0} (r3)
            (q2) edge pic {acc=1} (s3)
            (r3) edge (q4)
            (s3) edge (q4)
            ;

  \node[mystate,above of=q0] (q) {$q$};
  \node[mystate,right of=q] (r) {$s$};
  \node[mystate,left of=q] (s) {$r$};
  \node[empty,above of=q,node distance=7mm] (init) {};

  \draw[->] (init) edge (q);
  \draw[->] (q) edge[bend left] pic[pos=0.35] {acc=1} (r);
  \draw[->] (q) edge[bend right] pic[pos=0.35] {acc=0} (s);
  \draw[->] (r) edge[bend left] (q);
  \draw[->] (s) edge[bend right] (q);

\end{tikzpicture}
  \end{minipage}
  \label{fig:not-all-paths-accepting}
  \end{wrapfigure}
  \noindent
  We will now construct an accepting path~$\pi$ in~$\dagw$.
  Note that not all paths in~$\dagg^i$ are necessarily accepting (consider the
  TELA and the run DAG in the right, with the acceptance condition
  $\accinfof{\tacc 0} \land \accinfof{\tacc 1}$; there are many non-accepting paths from
  $(q,0)$---e.g., a~path that alternates between a~$q$-vertex and an~$r$-vertex
  and never touches any $s$-vertex).
  While constructing~$\pi$, for every clause~$\varphi_j$ we will be tracking the
  information about which atom of~$\varphi_j$ we should see next in order to
  satisfy~$\varphi_j$ on the path.
  In particular, we will start from a~vertex~$v_0$ that is a~root vertex
  of~$\dagg^i$ and we will use the tuple $t_0 = (\taccgof{c^1_1}, \ldots,
  \taccgof{c^\ell_1})$ to keep track of the colours.
  Using (Reach), it follows that there is a~clause~$\varphi_j$ s.t.~$v_0$ can
  reach a~$\taccgof{c^j_1}$-edge~$e_1$.
  We will set $t_1 = (\taccgof{c^1_1}, \ldots, \taccgof{c^j_2}, \ldots,
  \taccgof{c^\ell_1})$ and continue in
  a~similar way: from every vertex we encounter, we use (Reach) to obtain
  an edge that is a~$\taccgof{c}$-edge for some~$\taccgof{c}$ in~$t_i$.
  In the case we need to increment some component of~$t_i$
  from~$\taccgof{c^j_{k_j}}$, we set the new value to~$\taccgof{c^j_1}$.
  The path~$\pi$ is then constructed as an infinite path that goes through the
  infinite sequence~$v_0, e_1, e_2, \ldots$
  Note that because the sequence~$v_0, e_1, e_2 \ldots$ is infinite, due to the
  pigeonhole principle there will be a~clause~$\varphi_j$ s.t.\ the sequence
  $t_0, t_1, \ldots$ infinitely often increments the $j$-th component and
  so~$\pi$ is accepting.
  From~$\pi$, we can now extract the accepting run of~$\aut$ on~$w$.
\qed
\proofend

\begin{lemma}\label{lem:ranking-termination}
\cref{alg:label-run-dag} always terminates with $i \leq 2n$.
\end{lemma}

\begin{proof}
	Consider a run DAG $\dagw$ for a word $w$. First observe that at the end of
  the main loop of \cref{alg:label-run-dag} (\lnref{ln:increment}),
  $\dagg^i$~has no finite vertices (all of them were removed).
  Due to \lnref{ln:remove_finite_init}, $\dagg^i$ at
	the beginning of the main loop (\lnref{ln:endangered}) also has no finite vertices.
  Let $\dagg_m^i$ be
	the DAG $(V^i_m, E^i \cap (V^i_m\times V^i_m))$ where $V^i_m = \{ (q,j) \in V^i
	\mid j \geq m \}$, i.e., the projection of~$\dagg^i$ from level~$m$ down,
  and $\width(\dagg_m^i)$ is the maximum number of vertices on any level of the
  run DAG below level~$m$, formally, $\width(\dagg_m^i) = \max\{|\{(q,j) : (q,j)
  \in V^i_m\}| : j \geq m\}$.
  From the definition of endangered mapping and the loop
  on \lnref{ln:step_one}, we have that if the condition on
  \lnref{ln:endangered} holds, there is some $m\in\omega$ s.t.\
  $\width(\dagg_m^{i+1}) < \width(\dagg_m^{i})$.
  This holds because if the mapping~$\eta$ is non-empty, then there is at least
  one infinite path~$\dagg^i$ that is all removed in the next step, i.e., from
  some level~$m$, the width of all levels below get decreased by at least one.
  If the condition on \lnref{ln:endangered} does not hold, the algorithm
  terminates and we are done.
  From the previous claim we have that in each successful iteration of the main
  loop, the width of~$\dagg^{i+2}$ in the limit is at most the one of
  $\dagg^{i}$ minus one.
  Since the maximum width of~$\dagof w$ is~$n$, then, if~$w \notin \langof
  \aut$, at latest $\dagg_m^{2n-1}$ is empty for some $m \in \omega$,
	and hence $\dagg^{2n}$ is empty and the algorithm terminates.
  \qed
\end{proof}

\begin{lemma}\label{lem:term-in-lang}
	If $\word \in\langof{\aut}$, then \cref{alg:label-run-dag} terminates with $\bot$.
\end{lemma}
\begin{proof}
  Consider some $\word \in\langof{\aut}$.
  Then, there is an accepting run $\rho$ on~$\word$ in~$\aut$.
  We have $(\rho_j, j) \in \dagw$ for all $j\in\omega$; we show
	that $(\rho_j, j)$ is not $M$-endangered for every $M \in \modelsminof \acccond$.
  The fact that no $\rho_j$ is finite follows from the fact that~$\rho$ is
	infinite.
  Observe that for each $M \in \modelsminof \acccond$, there is some $\taccgof c\in M$
	s.t.\ $\taccgof c \in \accinfof{\rho}$ (otherwise, $\word$ would not be
  accepted by~$\aut$).
  Therefore, $(\rho_j, j)$ is not \mbox{$M$-endangered.}
  Hence, in every iteration of \cref{alg:label-run-dag}, all
	vertices $(\rho_j, j)$ stay in~$\dagg^i$.
  From \cref{lem:ranking-termination} we have that \cref{alg:label-run-dag}
  always terminates, but $\dagg^i \neq \emptyset$ for each~$i$.
  Therefore, the algorithm terminates with~$\bot$.
  \qed
\end{proof}

%

\begin{corollary}\label{lem:inf-label}
	$\word \notin \langof{\aut}$ iff \cref{alg:label-run-dag} on
	$\dagw$ terminates with $(r,m)$.
\end{corollary}
\begin{proof}
  $(\Rightarrow)$ follows from \cref{lem:bot-no-accept} by contraposition and
  $(\Leftarrow)$ follows from \cref{lem:term-in-lang} by contraposition.
  \qed
\end{proof}

\noindent
The following lemma about the ranking procedure will be useful later.

\begin{lemma}
\label{lem:rank}
If \cref{alg:label-run-dag} terminates with $(r,m)$, then $\max(r) \leq 2n$
and, moreover, either $\max(r) = 0$ or $r$ is tight.
\end{lemma}
\begin{proof}
The first part ($\max(r) \leq 2n$) follows directly from \cref{lem:ranking-termination}.
For the second part, there are two options: either~$\dagw$ is finite (i.e., there is no
infinite run of~$\aut$ on~$\word$), in which case \cref{alg:label-run-dag}
  assigns all vertices in~$\dagw$ rank~$0$ and does not even enter the loop at
  \lnref{ln:main_loop}.
In the other case ($\dagg$~is infinite), let $k = \max(r)$ if $\max(r)$ is odd
and $k = \max(r) - 1$ otherwise (from the previous case, we know that $k\geq 1$).
We know that for every $\ell \in \{1, 3, \ldots, k\}$, there is a~vertex
$v_\ell = (q_\ell, i_\ell) \in \dagw$ with $r(v_\ell) = \ell$ (this is because
the mapping at \lnref{ln:endangered} in the algorithm needs to be non-empty)
and that such a~vertex is the beginning of an infinite path of vertices with
rank~$\ell$.
Therefore, there needs to be a~level~$i$ containing vertices with all
ranks~$\{1, 3, \ldots, k\}$.
From the previous, all levels~$j > i$ will also have all of the odd ranks up
to~$k$.
Choosing~$i$ large enough will prevent level~$i$ having a~vertex with an even rank higher than~$k$.
Therefore, $r$~is tight.
%
%
%
%
%
\qed
\end{proof}


%
%
%

\vspace{-0.0mm}
\subsection{$\accinf$-TELA Complement Construction}
\label{subsec:infela}
\vspace{-0.0mm}


Let $\aut = (\states, \trans, \inits, \colourset, \colouring, \acccond)$ be an
$\accinf$-TELA and $n = |\states|$.
We define a~\emph{(level) ranking} to be a function
$f\colon \states \to \{0, \ldots, 2n\}$.
The \emph{rank} of $f$ is defined as $f = \max\{f(q) \mid q \in \states\}$. 
We call a~mapping $\mu\colon Q\to \modelsminof \acccond$ a~\emph{level model}.
We say that $\mu$ is \emph{consistent} wrt $f$ if 
\begin{inparaenum}[(i)]
  \item $\mu(q) \in \modelsminof \acccond$ if $f(q)$ is odd, and 
  \item $\mu(q) = \lexof{\overline \acccond}$ if $f(q)$ is even. 
\end{inparaenum}
We denote the set of all level models by~$\levelmodels$.
%
For a~set of states $S \subseteq \states$ and a level model $\mu$, we call~$f$ to be $\Smutight$ if
%

\begin{center}
  \begin{tabular}{ll}
    (i)~it has an odd rank~$r$,\hspace{20mm} &
    (ii)~$f(S) \supseteq \{1, 3, \ldots, r\}$, \\
    (iii)~$f(\states\setminus S) = \{0\}$, and &
	  (iv)~$\mu$ is consistent wrt\ $f$.
  \end{tabular}
\end{center}

\noindent
A~ranking is $\mutight$ if it is $(\states,\mu)$-tight; we use $\cT$ to denote
the set of all $\mu$-tight rankings for some level model~$\mu$.

For two level rankings $f, f'$ and two level models~$\mu, \mu'$, we say
that~$(f', \mu')$ is a~\emph{consistent successor} of~$(f,\mu)$ over~$a$,
denoted as $(f, \mu) \fsucc (f', \mu')$, iff
\begin{enumerate}[(i)]
  \item $\mu$ and $\mu'$ are consistent wrt $f$ and $f'$, respectively, and
  \item for all $q \in \domof{f}$ and $q'\in\trans(q,a)$ the following holds: 
  \begin{enumerate}[(a)]
    \item $f'(q') \leq f(q)$, 
    \item $(\colouring(q \ltr a q') \cap \mu(q) \neq \emptyset) \Rightarrow 
    f'(q') \leq \evenceil{f(q)}$, and 
    \item $\mu'(q') \neq \mu(q) \Rightarrow f'(q') \leq \evenceil{f(q)}$.
  \end{enumerate}
\end{enumerate}
Intuitively, the rankings guess the ranks of states in the run DAG and the level models guess the models assigned to states in the labelling procedure described in~Section~\ref{sec:infelalabelling}. 
Consistent successors respect the labelling procedure. On every path in a~run DAG, the ranks are nonincreasing. 
If some vertex $v$ with an odd rank has an~outgoing $\taccgof{c}$-edge to $v'$
and $\taccgof{c}$ is in the model assigned to $v$, the vertex~$v'$ has to have
a~lower rank than~$v$, because when~$v$ is removed from $\dagw^i$, there is no
reachable $\taccgof{c}$-edge in $\dagw^i$. 
Moreover, if the model is changed between $v$ and $v'$, then the rank also has to be decreased. 

%
The complement of~$\aut$ is given as the BA
$\alginfela(\aut) = (\states', \trans', \inits', \acc')$ whose components are defined as
follows:

\begin{itemize}
  \item  $\states' = \states_1 \cup \states_2$ where
    \begin{itemize}
      \item  $\states_1 = 2^\states$ and
      \item  $\states_2 = \hspace*{-1mm}
        \begin{array}[t]{ll}
          \{(S,O,f,i,\mu) \in \hspace*{-0mm}& 2^\states \times 2^\states \times
          \cT \times \{0, 2, \ldots, 2n - 2\}\times \levelmodels \mid {} \\
          & f \text{ is } (S, \mu)\textrm{-tight}, O \subseteq S \cap f^{-1}(i)\},
        \end{array}$
    \end{itemize}
  \item  $\inits' = \{\inits\}$,
  \item  $\trans' = \trans_1 \cup \trans_2 \cup \trans_3$ where
    \begin{itemize}
      \item  $\trans_1\colon \states_1 \times \Sigma \to 2^{\states_1}$ such that $\trans_1(S, a) =
        \{\trans(S,a)\}$,
      \item $\trans_2\colon \states_1 \times \Sigma \to 2^{\states_2}$ s.t.\ $\trans_2(S, a) =
        \{(S', \emptyset, f, 0, \mu) \mid S' = \trans(S, a)\}$, and
      \item $\trans_3\colon \states_2 \times \Sigma \to 2^{\states_2}$ such that $(S', O', f', i', \mu') \in
        \trans_3((S, O, f, i,\mu), a)$ iff
          \begin{itemize}
            \item  $S' = \trans(S, a)$,
            \item  $(f, \mu) \fsucc (f', \mu')$,
            \item  $\rankof f = \rankof{f'}$,
            \item  and
              \begin{itemize} 
                \item  $i' = (i+2) \mod (\rankof{f'} + 1)$ and $O' = f'^{-1}(i')$ if
                  $O = \emptyset$ or
                \item  $i' = i$ and $O' = \trans(O, a) \cap f'^{-1}(i)$ if $O \neq
                  \emptyset$, and
              \end{itemize}
          \end{itemize}
    \end{itemize}
  \item  $\acc' = \{\emptyset \ltr a \emptyset \in \trans_1 \mid a \in \Sigma\} \cup 
    \{M_1 \ltr a M_2 \in \trans_3 \mid M_1 = (\cdot, \emptyset, \cdot, \cdot,
    \cdot), a \in \Sigma\}$
\end{itemize}

\noindent
Intuitively, a~run of $\alginfela(\aut)$ on a~word~$\word$ tries to construct
the run DAG~$\dagw$ of~$\aut$ on the same word, with rankings encoded within
the states.
The restrictions on~$\delta_3$ encode the rules from \cref{alg:label-run-dag}.
The partitioning of~$\states'$ into~$\states_1$ and~$\states_2$ allows us to
consider only tight rankings, as in~\cite{FriedgutKV06}.
Moreover, the $i$-component of a~macrostate allows us further decrease the
number of states in the same way as in~\cite{Schewe09} (we know that all states
in~$O$ have the same rank~$i$).

\begin{theorem}
	Let $\aut$ be an $\accinf$-TELA. Then, $\langof{\alginfela(\aut)} = \Sigma^\omega \setminus \langof{\aut}$.
\end{theorem}
\begin{proof}
  ($\subseteq$)
  We use Boolean laws and prove an equivalent statement $\langof{\aut} \subseteq
  \Sigma^\omega \setminus \langof{\alginfela(\aut)}$.
  Let $\word \in \langof{\aut}$ be a~word and~$\rho$ be an accepting run
  of~$\aut$ on $\word$. 
  First, let~$\rho'$ be the run $\rho' = S_0 S_1\ldots$ with $S_0=I$ and
  $S_{i+1}=\delta_1(S_i,\word(i))$ for all $i \in \omega$ (i.e., $\rho'$
  stays in~$\states_1$).
  The run~$\rho'$ cannot be accepting in $\alginfela(\aut)$,
  because $\rho(i) \in S_i$ and so $S_i \neq \emptyset$ for any $i \in \omega$
  (in~$\states_1$, the only accepting transitions are those from
  state~$\emptyset$ to state~$\emptyset$). 
  Second, let
  \vspace{-2mm}
  \begin{align*}
  \rho'' = S_0 S_1 \ldots S_p(S_{p+1},O_{p+1}, f_{p+1},i_{p+1},\mu_{p+1})(S_{p+2},O_{p+2},f_{p+2},i_{p+2},\mu_{p+2})\ldots\\[-7mm]
  \end{align*}
  be a~run of $\alginfela(\aut)$ on $\word$ ($\rho''$ jumps to~$\states_2$ at
  position~$p$).
  From the construction, it holds that $(f_{j}, \mu_{j}) \fsucc (f_{j+1},
  \mu_{j+1})$ for all $j>p$.
  Since $\rho$ is accepting in $\aut$, eventually there will be a~position $k > p$ such that
  $f_k(\rho(k))$, $f_{k+1}(\rho(k+1))$, $f_{k+2}(\rho(k+2))$, \ldots are all even (because there is
  no model satisfying $\rho$ in $\modelsminof \acccond$, so points (iib) and
  (iic) from the definition of~$\fsucc$ will enforce this).
  For the sake of contradiction, assume that~$\rho''$ is accepting.
  Then for some position $\ell > k$, because the~$i$-component of a~macrostate
  rotates over all even ranks, it holds that $i_\ell = f_\ell(\rho(\ell))$ and
  $\rho(\ell) \in O_\ell = f_\ell^{-1}(\rho(\ell))$. 
  We can easily show by induction that for all $j \geq \ell$, it holds that
  $\rho(j) \in O_j \neq \emptyset$, which is in contradiction with the
  assumption that~$\rho''$ is accepting. 


  ($\supseteq$)
  Consider any word $\word \not\in \langof{\aut}$.
  From~\cref{lem:inf-label,lem:rank} it follows that the run DAG $\dagw$ has a~bounded rank.
  If all vertices of $\dagw$ are finite, then there is an accepting run $\rho'$ on $\alginfela(\aut)$ where $\rho' = S_0 S_1 \ldots$ with $S_0 = I$ and $S_{i+1} = \delta(S_i, \word_i)$ for all $i \in \omega$. 
  Otherwise, \cref{alg:label-run-dag} terminates with a~tight ranking~$r$ and a~model~$m$.
  From the definition of~$\fsucc$, there is a~run
  \vspace{-2mm}
  \begin{align*}
  \rho'' = S_0S_1 \ldots S_p(S_{p+1},O_{p+1},f_{p+1},i_{p+1},\mu_{p+1})(S_{p+2},O_{p+2},f_{p+2},i_{p+2},\mu_{p+2})\ldots\\[-7mm]
  \end{align*}
  such that $f_k(q) = r((q,k))$ and $\mu_k(q) = m((q,k))$ for all $k > p$.
  In order to show that $\rho''$ is acepting, we need to show that
  the $O$-component of the macrostates on the run is empty infinitely often.
  Assume by contradiction that there is an index~$\ell > p$ such that~$O_j$ is
  non-empty for all $j \geq \ell$.
  Then, there is a vertex $(q,\ell) \in\dagw$ s.t.\ $r((q,\ell))$ is even and
  there are infinitely many vertices reachable from $(q,\ell)$ with the same even
  rank, which is a contradiction with the construction of~$r$ in
  \cref{alg:label-run-dag}, which would give some of the vertices odd ranks.
  \qed
\end{proof}

For the complexity analysis, we use $\tight(n)$ to denote the number of
$\mu$-tight level rankings for an automaton with $n$ states ($\mu$-tight
rankings for $\accinf$-TELAs correspond to \emph{tight} rankings for BAs).
It holds that $\tight(n) \approx (0.76n)^n$~\cite{FriedgutKV06,Schewe09}.

\begin{theorem}\label{thm:inf-tela-complexity}
The number of states of $\alginfela(\aut)$ is in $\bigOof{k^n\cdot
\tight(n+1)} = \bigOof{n (0.76nk)^n}$ for $k = |\modelsminof \acccond|$.
\end{theorem}

\begin{proof} 
The set of macrostates $\states_1$ is obtained by a~simple subset construction,
therefore $\states_1 \in \bigOof{2^n}$.
That is much smaller than $\bigOof{k^n\cdot \tight(n+1)}$, so it is sufficient to count only the number of macrostates of the form $(S, O, f, i, \mu)$.
For this, we uniquely encode each macrostate as a~pair $(h, i)$ where $h \colon \states \rightarrow \{ -2, -1, \ldots, 2n-1 \} \times \modelsminof \acccond$ is defined as follows:
 \begin{equation}
  h(q) = \begin{cases}
    (-1, \mu) & \text{if } q \in O, \\
    (-2, \mu) & \text{if } q \in \states\setminus S, \\
    (f(q), \mu) & \text{otherwise}.
  \end{cases}
\end{equation}
We compute the number of encodings $h$ for a~fixed $i$.
We divide all encodings into four groups according to the set $\img(h)_0 \cap \{ -2,
-1 \}$ where $\img(h)_0$ denotes the set of first elements of the pairs in
$\img(h)$.
We show that we can obtain the bound
$\bigOof{k^n \cdot \tight(n)}$ for each of the groups.
The groups are denoted by $g_M$ with $M \subseteq \{-2, -1\}$.
For $h(q) = (m, \mu)$, we use $h(q)_m$ \mbox{and $h(q)_\mu$ to denote~$m$ and~$\mu$.}

  \begin{enumerate}
    \item[$g_\emptyset$:] from the definition, $f$~is $\mu$-tight.
      The level model $\mu$ is of the form $\mu \colon \states \rightarrow
      \modelsminof \acccond$, so there are $k$ possible assignments for every
      state from $\states$.
      The number of level models is therefore $k^n$ and $|g_\emptyset| =
      \bigO(k^n \cdot \tight(n))$. 
    
    \item[$g_{\{-1\}}$:] since there is at least one state $q$ with $h(q)_m = -1$,
      this means that $q \in O$ so~$q$ has an even rank.
      As a~consequence, at least one of the positive odd ranks of~$h$ (up to
      $2n-1$) will not be taken, so we can infer that $h\colon \states \to
      \{-1, \ldots, 2n-3\} \times \modelsminof \acccond$.
      We can therefore uniquely represent~$h$ by a~mapping~$h'$ by incrementing all
      ranks of~$h$ by two, so $h'\colon \states \to \{0, \ldots, 2n-1\} \times \modelsof \acccond$.
      But then $h' \in \cT(n)$ and the number of all level models is $k^n$, so $|g_{\{-1\}}| \in \bigO(k^n \cdot \tight(n))$.

    \item[$g_{\{-2,-1\}}$:] similarly as for $g_{\{-1\}}$ we get
      that $|g_{\{-2,-1\}}| \in \bigO(k^n \cdot \tight(n))$.

    \item[$g_{\{-2\}}$:] the reasoning is similar to the one for  $g_{\{-1\}}$,
      with the exception that now, we know that there is a~state~$q \in \states 
      \setminus S$, which is, according to the definition of a~ranking, assigned
      the rank~$0$.
      This means that one positive odd rank of~$h$ is, again, not taken, so we
      increment all non-negative ranks of~$h$ by two and map all states in $\states
      \setminus S$ to~$1$, obtaining a~tight ranking $h' \in \cT(n)$.
      The number of level models is $k^n$, 
      therefore, $|g_{\{-2\}}| \in \bigO(k^n \cdot \tight(n))$.
  \end{enumerate}

  \noindent
  Since the size of all groups is bounded by $\bigO(k^n \cdot \tight(n))$, for
  a~fixed~$i$, the total number of these encodings is still
  $\bigO(k^n \cdot \tight(n))$.
  When we sum the encodings for all~$i$'s, we obtain that the number is
  bounded by~$\bigO(k^n \cdot \tight(n+1))$, since $\bigOof{n \cdot \tight(n)}
  = \bigOof{\tight(n+1)}$~\cite{Schewe09}.
  The rest follows from the
  approximation of $\tight(n)$.
\qed
\end{proof}

\begin{corollary}\label{cor:inf-tela-compl-size}
  Let~$\aut$ be an $\accinf$-TELA with~$n$ states and~$k$ colours~$\colourset$.
  The number of states of $\alginfela(\aut)$ is in
  $\bigOof{{k \choose \lfloor k/2 \rfloor}^{n} \cdot \tight(n+1)} =
      \bigOof{n \cdot ({k \choose \lfloor k/2 \rfloor} \cdot 0.76n)^n}
  \subseteq
      \bigOof{n (2^k \cdot 0.76n)^n}$.
\end{corollary}

\begin{proof}
The proof of the more precise bound follows directly from
\cref{thm:inf-tela-complexity} and the fact that the size of $\modelsminof
\acccond$ is bounded by the size of the largest antichain in~$2^\colourset$,
which is at most ${k \choose \lfloor k/2 \rfloor}$ by Sperner's theorem.
\qed
\end{proof}

\begin{corollary}
  Let~$\aut$ be a~GBA with~$n$ states and~$k$~colours.
  Then the number of states of $\alginfela(\aut)$ is in $\bigOof{k^n \cdot
  \tight(n+1)} = \bigOof{n (0.76nk)^n}$.
\end{corollary}

\begin{proof}
The proof follows directly from \cref{thm:inf-tela-complexity}.
For a~GBA it holds that $\overline{\acccond} = \bigvee_{0 \leq j < k} \taccj$. The formula is in DNF, hence $\modelsminof \acccond = \{ \{ \taccj \} \mid 0 \leq j < k \}$ and $|\modelsminof \acccond| = k$. The number of all level models is $k^n$. 
The rest of the proof is done in the same way as in the proof of \cref{cor:inf-tela-compl-size}.
\qed
\end{proof}

We note that to the best of our knowledge, our bound on the complementation of
GBAs is better than other bounds in the literature.
In particular, it is clearly better than the bound $\bigO(k^n(2n + 1)^n)$
from~\cite{KupfermanV05}, which is the best upper bound for complementing GBAs
that we are aware of.
It is also better than an approach that would go through determinization by
using the procedure in~\cite{ScheweV12}, which outputs a~deterministic Rabin
automaton with at most $\mathsf{ghist}_k(n)$ states and $2^n- 1$ accepting
pairs, which can be complemented easily into a~Streett automaton.
According to~\cite{ScheweV12} $\mathsf{ghist}_k(n)$ converges to $(1.47nk)^n$
for large $k$, which is already worse than our upper bound.
~

\DeclareRobustCommand\safecol[1]{\taccgof{#1}}
\vspace{-0.0mm}
\section{Modular Complementation of $\accfinof{\safecol c} \land \varphi$ TELAs}\label{sec:fin-modular}
\vspace{-0.0mm}

In this section, we propose a~modular algorithm $\instancompl$ for complementation of
TELAs with the acceptance condition $\accfinof{\taccgof c} \land \varphi$ for
any~$\varphi$, parameterized by an~algorithm for complementing TELAs
with the condition~$\varphi$.
In \cref{sec:instantiations}, we will then instantiate the algorithm for some common
acceptance conditions, eventually obtaining an efficient complementation algorithm for general TELAs.

Let us fix a~TELA $\aut = (\states, \delta, I, \colourset, \colouring,
\accfinof{\taccgof c} \land \varphi)$ and let $\Delta$ be~$\delta$ without
transitions whose label contains~$\taccgof c$.
For a~word $\word \in \Sigma^\omega$, we define a~\emph{relaxed run DAG} (RRDAG) over $\word$, 
denoted by $\gendag$, as 
any sequence of states $\gendag = (S_0, S_1, \dots)$ where $S_i \subseteq \states$ and 
$\Delta(S_{i}, \wordof{i}) \subseteq S_{i+1}$. Intuitively, an~RRDAG over a~word may contain 
more states on each level than it is necessary from the reachability of~$\Delta$. 
Note that this definition of RRDAGs is equivalent to having vertices of the form $(q,i)$, where 
$q \in S_i$ with edges given implicitly by $\Delta$.
We~use these definitions interchangeably.
Clearly, there may be multiple RRDAGs over a~single word, they are all, however,
subgraphs of the (standard) run DAG~$\dagw$.
We say that $\gendag = (S_0, S_1, \dots)$ is \emph{accepting} wrt $\varphi$,
written as $\gendag \models \varphi$, if there is a~run $\rho = q_k q_{k+1} \ldots$
for $k \geq 0$
in~$\Delta$ such that for every $i \geq k$ it holds that $q_i \in S_i$ and
$q_{i+1} \in \Delta(q_i, \wordof{i})$, and, moreover, $\rho \models \varphi$
(i.e., the accepting run does not need to start at the beginning of $\gendag$).
The reason for introducing RRDAGs is that the algorithm for
condition~$\varphi$ will construct a~BA that runs over RRDAGs constructed using
the restricted transition relation~$\Delta$.
The relaxation allows us to introduce new vertices (not connected to the root
of the RRDAG) at any level that represent runs that have seen finitely many
times a~$\taccgof c$ transition in~$\delta$.

Our definition of the modular procedure $\instancompl$ for
$\accfinof{\taccgof c} \land \varphi$ is given wrt
a~\emph{subprocedure} for complementing a~TELA with condition~$\varphi$.
The subprocedure is given as a~tuple $\instanDeltavarphi = (\M, \M_0,
\succact_{\Delta}, \succtrack_{\Delta}, \breakempty)$, where

%
\begin{enumerate}[(i)]
  \item $\M$ is a set of \emph{macrostates},
  \item $\M_0 \subseteq \M$ is a set of \emph{initial macrostates},
  \item $\succact_{\Delta}\colon 2^\states \times \Sigma \times \M \to 2^\M$ is
    an \emph{active transition function},
  \item $\succtrack_{\Delta}\colon 2^\states \times \Sigma \times \M \to 2^\M$
    is a \emph{tracking transition function}, and
  \item $\breakempty \subseteq \M$ is an \emph{empty-breakpoint} predicate.
\end{enumerate}
We use $\succunion_\Delta$ to denote $\succact_\Delta \cup \succtrack_\Delta$
(when treated as relations).
Intuitively, $\M$ is a~set of macrostates given by the subprocedure
for~$\varphi$.
$\breakempty$ is a~condition that has to hold for a~macrostate to be accepting in~$\instanDeltavarphi$.
The transitions between macrostates of $\M$ are described using transition
functions $\succact_\Delta$ and $\succtrack_\Delta$.
In particular, $\mst' \in \succunion_\Delta(P', a, \mst)$ is computed by taking
the successor of the macrostate~$\mst$ over~$a$, but also while taking into
account the set~$P'$ of states ($\mst$ corresponds to index~$i$ of the run
while~$\mst'$ and~$P'$ correspond to index~$i+1$) provided by $\instancompl$,
which represent breaking the~$\accfinof{\taccgof c}$ condition.
The reason for using two transition functions ($\succact_\Delta$ and
$\succtrack_\Delta$) is that some subprocedures that
we will introduce later will use two types of macrostates: active and tracking.
For instance, if $\instanDeltavarphi$ is a~rank-based procedure (cf.\
\cref{sec:rabin}), active macrostates will contain breakpoint, which the
construction will try to empty, and once a~breakpoint is seen, $\instancompl$
will add some more runs to the rank-based algorithm.
The new runs might not be tight at the given point, so we switch into the
tracking mode and wait for newly added runs to become tight before switching
into the active mode again.

Let $\word$ be a word and $\gendag = (S_0, S_1, \dots)$ be an~RRDAG over~$\word$.
A \emph{\finrun} $\instanrun$ of $\instanDeltavarphi$ over $\gendag$ is a~sequence 
$(\mst_0, \mst_1, \dots)$ s.t.\ $\mst_0 \in \M_0$ and $\mst_{i+1} \in \succunion_\Delta(S_{i+1}, \wordof{i}, \mst_i)$ for all $i \geq 0$.
$\instanrun$ is \emph{accepting} if $\breakempty(\mst_i)$ holds for
infinitely many $i$'s.
We say that the subprocedure~$\instanDeltavarphi$ is \emph{correct for~$\varphi$} if for each
word $\word$ and every RRDAG~$\gendag$ over~$\word$ it holds that $\gendag$ is
not accepting wrt~$\varphi$ iff there is an accepting \finrun~$\instanrun$ of
$\instanDeltavarphi$ over~$\gendag$.

Let us now move to the definition of $\instancompl$.
For subprocedure~$\instanDeltavarphi$ and TELA~$\aut$ given above, the
algorithm  will construct the BA $\instancompl(\instanDeltavarphi, \aut)=
(\states', I', \delta', F')$ defined as follows:

\begin{itemize}
  \item  $\states' = \hspace*{-1mm}
  \begin{array}[t]{ll}
    \{(S,P,\mst) \in 2^\states \times 2^\states \times \M
    \},
  \end{array}$
  \item  $I' = \{(I, I, \mst_0) \mid \mst_0 \in \M_0\}$,
  \item  $\delta' = \delta_1 \cup \delta_2$ where
    \begin{itemize}
      \item $\delta_1\colon \states' \times \Sigma \to 2^{\states'}$ such that $(S', P', \mst') \in
        \delta_1((S, P, \mst), a)$ iff
          \begin{itemize}
            \item $S' = \delta(S, a)$,
            \item if $\breakempty(\mst)$: $P' = S'$,
            \item if $\neg \breakempty(\mst)$: $P' = \Delta(P, a)$,
            \item $\mst' \in \succact_{\Delta}(P', a, \mst)$,
          \end{itemize}
      \item $\delta_2\colon \states' \times \Sigma \to 2^{\states'}$ such that $(S', P', \mst') \in
        \delta_2((S, P, \mst), a)$ iff
          \begin{itemize}
            \item  $S' = \delta(S, a)$,
            \item  $P' = \Delta(P, a)$,
            \item $\mst'\in \succtrack_{\Delta}(P', a, \mst)$, and
          \end{itemize}
    \end{itemize}
  \item  $F' = \{  (S, P, \mst) \transover{a} (S', P', \mst') \in \delta' \mid a\in\Sigma, \breakempty(\mst') \}$.
\end{itemize}
Intuitively, the construction executes~$\instanDeltavarphi$ on the restricted
transition relation~$\Delta$, while also keeping track of all runs~(in~$S$) and
runs that either need to terminate or see a~$\taccgof c$-transition~(in~$P$).
Whenever~$\instanDeltavarphi$ clears its breakpoint, $P$~is re-sampled (and
some new runs can be added to~$\instanDeltavarphi$). 

\begin{restatable}{theorem}{thmModCorrect}\label{thm:mod_correct}
  For a~correct subprocedure $\instanDeltavarphi$,
  $\langof{\instancompl(\instanDeltavarphi, \aut)} = \Sigma^\omega \setminus \langof{\aut}$.
\end{restatable}
  

\noindent
The overhead of the procedure over the subprocedure $\instanDeltavarphi$ is at most~$3^n$-times.

\begin{theorem}\label{thm:mod_complexity}
Suppose $\instanDeltavarphi =(\M, \cdot, \cdot, \cdot)$.
Then $|\instancompl(\instanDeltavarphi, \aut)| \in \bigO(3^n \cdot |\M|)$.
\end{theorem}
\begin{proof}
Since in $(S, P, \mst)$, it always holds that $P \subseteq S$, each state of~$\aut$ can be in one of the three following sets:
\begin{inparaenum}[(i)]
  \item  $\states \setminus S$,
  \item  $S \cap P$, and
  \item  $S \setminus P$.
\end{inparaenum}
\qed
\end{proof}

\vspace{-0.0mm}
\section{Complementation of TELAs and their Subclasses}\label{sec:instantiations}
\vspace{-0.0mm}

We proceed by instantiating the modular algorithm $\instancompl$ from the
previous section for several common automata classes---co-\buchi automata,
Rabin automata, parity automata, generalized Rabin automata, and, eventually,
TELAs.


\subsection{Co-\buchi Automata}

As a~simple demonstration of instantiation of $\instancompl$, we use it to create
a~complementation algorithm for co-\buchi automata.
The acceptance condition for co-\buchi automata is $\accfinof{\tacc 0} =
\accfinof{\tacc 0} \land \mytrue$, we therefore need to provide a~trivial
subprocedure $\instantrue =(\mintrue, \mintrue_0, \succacttrue_\Delta,
\emptyset, \breakemptytrue)$ that is correct for~$\mytrue$ (notice that $\succtracktrue_\Delta$ is empty).
In the subprocedure, $\mintrue = 2^\states$, $\mintrue_0 = \{ I \}$, and 
the remaining components are given as follows:
%
\begin{equation*}
  \succacttrue_\Delta(P, a, S) = \{P\} \qquad\text{and}\qquad
  \breakemptytrue(P) \Longleftrightarrow P = \emptyset.
\end{equation*}
Intuitively, the instantiated procedure works with macrostates~$(S,P,P)$ (i.e.,
to adhere to the formal definition of $\instancompl$, $P$
is there twice) where~$S$ tracks all runs and~$P$ is a~breakpoint that contains
runs that yet need to either terminate or see $\tacc 0$.
To accept, $P$ needs to be emptied infinitely often.
One can observe that $\instancompl(\instantrue, \aut)$ resembles the well-known
Miyano-Hayashi construction~\cite{MiyanoH84} for complementation of
co-B\"{u}chi automata.

\begin{restatable}{lemma}{lemInstantTrueCorrect}\label{lem:instantrue-correct}
  The subprocedure $\instantrue$ is correct for the acceptance condition $\mytrue$.
\end{restatable}

\begin{corollary}
For a~co-\buchi automaton~$\aut$, $\langof{\instancompl(\instantrue, \aut)} = \Sigma^\omega \setminus \langof \aut$.
\end{corollary}
\begin{proof}
Follows from \cref{lem:instantrue-correct,thm:mod_correct}.
\qed
\end{proof}

\noindent
Since the result of the construction can be mapped to the Miyano-Hayashi's
algorithm~\cite{MiyanoH84}, the complexities also match.
%
\begin{corollary}
  $|\instancompl(\instantrue, \aut)| \in \bigO(3^n)$
\end{corollary}

\subsection{Rabin Automata}
\label{sec:rabin}

In this section, we give an instantiation of $\instancompl$ with subprocedure
$\instinf = (\minf$, $\minf_0$, $\succactinf_\Delta$,
$\succtrackinf_\Delta, \breakemptyinf)$ for $\accinfof{\tacc 1}$, which will allow us to complement TELAs
where the acceptance condition is a~single Rabin pair.
The algorithm is based on the optimal rank-based BA complementation algorithm
from~\cite{Schewe09} adjusted to the needs of the modular construction.
The macrostates of the instantiation \mbox{are given as}
\begin{equation*}
\minf = 
  \overset{\minfact}{\overbrace{2^\states \cup (\cT\times 2^\states \times \{ 0, 2, \ldots, 2n - 2 \})}}
  \cup
  \overset{\minftrack}{\overbrace{(\cT\times \{ 0, 2, \ldots, 2n - 2 \})}}
\end{equation*}
where $\minf_0 = \{ I \}$.
Notice that \emph{active macrostates} ($\minfact$) are either sets of states
(from~$2^\states$, just keeping track of all runs) or states of the form $(f,
O, i)$ (representing tight runs).
On the other hand, \emph{tracking macrostates} ($\minftrack$) are of the form
$(f, i)$; these are used to wait for newly arrived runs to become tight.
The remaining components are then defined as follows:

\noindent
\begin{minipage}{\textwidth}\scriptsize
\begin{multicols}{2}
\begin{itemize}
  \item $(f', O', i') \in \succactinf_\Delta(P, a, (f, O, i))$ iff
  \begin{itemize}
    \item  $f \finsucca{\Delta} f'$ and $\rankof f = \rankof{f'}$,
    \item $\domof{f'} = P$,
    \item $O \neq \emptyset$, 
    \item $i' = i$, 
    \item $O' = \Delta(O, a) \cap f'^{-1}(i)$
  \end{itemize}
  \item $(f', i') \in \succactinf_\Delta(P, a, (f, O, i))$ iff
  \begin{itemize}
    \item  $f \finsucca{\Delta} f'$ and $\rankof f = \rankof{f'}$,
    \item $O = \emptyset$,
    \item $i' = (i+2) \mod (\rankof{f'} + 1)$
  \end{itemize}

  \item $P' \in \succactinf_\Delta(P, a, P)$ iff
  \begin{itemize}
    \item $P' = P$
  \end{itemize}
  \item $(f', i') \in \succtrackinf_\Delta(P, a, P)$ iff
  \begin{itemize}
    \item $f'$ is $P$-tight
    \item $i'= 0$
  \end{itemize}
  \item $\{(f', i'), (f', O', i')\} \subseteq \succtrackinf_\Delta(P, a, (f, i))$ iff
  \begin{itemize}
    \item $f \finsucca{\Delta} f'$ and $\rankof f = \rankof{f'}$,
    \item $O' = f'^{-1}(i')$,
    \item $i'= i$
  \end{itemize}
  \item $\breakemptyinf((f,O,i)) \Longleftrightarrow O = \emptyset$
  \item $\breakemptyinf(P) \Longleftrightarrow P = \emptyset$
  \item $\breakemptyinf((f,i))\Longleftrightarrow \mathit{false}$
\end{itemize}
\end{multicols}
\end{minipage}
\vspace{2mm}

\noindent
An example of the construction is shown in \cref{acc:rabin-example}.
The correctness of the instantiation is then summarized by the following lemma.
\begin{restatable}{lemma}{lemInstantInfCorrect}
The subprocedure $\instinf$ is correct for the acceptance condition $\accinfof{\tacc 1}$.
\end{restatable}

\noindent
The following lemma shows that using our approach, handling the
$\accfinof{\taccgof c}$ condition is ``\emph{for free},'' i.e., the
asymptotical complexity stays the same as for the optimal algorithm for BA
complementation from~\cite{Schewe09}.

\begin{lemma}
	$|\instancompl(\instinf, \aut)| \in  \bigO(\tight(n+1))$.
\end{lemma}
\begin{proof}
  It suffices to count the number of macrostates of the form $(S,P,f,O,i)$. 
  Consider a macrostate $(S,P,f,O,i)$. We uniquely encode the macrostate as
  $(h,i)$ where $h\colon \states \to \{-3, \ldots, 2n-1\}$ is defined as follows:
	\begin{equation}
		h(q) = \begin{cases}
			-1 & \text{if } q \in O, \\
			-2 & \text{if } q \in \states\setminus S, \\
			-3 & \text{if } q \in S \setminus P, \text{ and} \\
			f(q) & \text{otherwise}.
		\end{cases}
	\end{equation}
	For a~fixed~$i$ we compute the number of such encodings~$h$.
  First we divide all encodings into groups according to the set $\imgof{h}
  \cap \{ -3, -2, -1 \}$ (8~groups at most) and we will show for each of the
  groups how we can ``\emph{shuffle}'' the ranks in~$h$ to obtain the bound
  $\bigO(\tight(n))$ for each of the groups.
  We will denote each of the groups by $g_M$ with $M \subseteq \{-3, -2, -1\}$.

  \begin{enumerate}
    \item[$g_\emptyset$:] from the definition, $f$~is tight so $|g_\emptyset| =
      \bigO(\tight(n))$
    \item[$g_{\{-1\}}$:] since there is at least one state $q$ with $h(q) = -1$,
      this means that $q \in O$ so~$q$ has an even rank.
      As a~consequence, at least one of the positive odd ranks of~$h$ will not
      be taken, so we can infer that $h\colon \states \to \{-1, \ldots, 2n-3\}$.
      We can therefore uniquely map~$h$ to a~mapping~$h'$ by incrementing all
      ranks of~$h$ by two, so $h'\colon \states \to \{1, \ldots, 2n-1\}$.
      But then $h' \in \cT(n)$, so $|g_{\{-1\}}| \in \bigO(\tight(n))$.

    \item[$g_{\{-2,-1\}}$:] via the same reasoning as for $g_{\{-1\}}$ we get
      that $|g_{\{-2,-1\}}| \in \bigO(\tight(n))$.

    \item[$g_{\{-2\}}$:] the reasoning is similar to the one for  $g_{\{-1\}}$,
      with the exception that now, we know that there is a~state~$q \in \states
      \setminus S$, which is, according to the definition of a~ranking, assigned
      the rank~$0$.
      This means that one positive odd rank of~$h$ is, again, not taken, so we
      increment all non-negative ranks of~$h$ by two and map all states in $\states
      \setminus S$ to~$1$, obtaining a~tight ranking $h' \in \cT(n)$.
      Therefore, $|g_{\{-2\}}| \in \bigO(\tight(n))$.

    \item[$g_{\{-3\}}$:] the reasoning is, again, similar to the one for  $g_{\{-1\}}$,
      with the exception that now, we know that there is a~state~$q \in S
      \setminus P$ such that its rank is, according to the definition~$0$.
      Therefore, we increment all non-negative ranks of~$h$ by two and map the
      states in $S \setminus P$ to~$1$, obtaining a~tight ranking $h' \in
      \cT(n)$;
      therefore, $|g_{\{-3\}}| \in \bigO(\tight(n))$.

    \item[$g_{\{-3,-2\}}$, $g_{\{-3, -1\}}$:]
      similarly as for $g_{\{-2\}}$, we increment all non-negative ranks of~$h$
      by two and set $h'(q) = 0$ if $h(q) = -3$ and $h'(q) = 1$ if $h(q) = -2$
      (resp.\ if $h(q) = -1$).
      Then $h' \in \cT(n)$ and so $|g_{\{-3,-2\}}| = \bigO(\tight(n))$ and
      $|g_{\{-3, -1\}}| \in \bigO(\tight(n))$.

    \item[$g_{\{-3, -2, -1\}}$:] in this case, we know that there is at least
      one state $q_1 \in O$ and at least one state $q_2 \in \states \setminus S$.
      Therefore, there will be at least two odd positions not taken in~$h$, so
      we can infer that $h\colon \{-3, \ldots, 2n-5\}$.
      We create $h'$ by incrementing all ranks in~$h$ by \emph{four}; in this
      way, we obtain a~tight ranking $h'\colon \states \to \{0, \ldots, 2n-1\}$, so
      $|g_{\{-3, -2, -1\}}| \in \bigO(\tight(n))$.
  \end{enumerate}

  \noindent
  Since the size of all groups is bounded by $\bigO(\tight(n))$, for
  a~fixed~$i$, the total number of these encodings is still
  $\bigO(\tight(n))$.
  When we sum the encodings for all possible~$i$'s, we obtain that the number is
  bounded by~$\bigO(\tight(n+1))$, since $\bigO(n \cdot \tight(n)) =
  \bigO(\tight(n+1))$~\cite{Schewe09}. 
  \qed
\end{proof}

The modular construction instantiated with $\instinf$ gives us a~procedure for 
the complementation of Rabin automata with a~single pair.
In order to get a~procedure for general Rabin automata, we construct
a~complement automaton for each Rabin pair and make a~product of these automata
and obtain a~GBA accepting the complement of the original Rabin automaton.
The complexity reasoning is then straightforward and is summarized by the
following corollary.

\begin{corollary}\label{theorem:rabin}
  Let $\aut$ be a~Rabin automaton with $k$~Rabin pairs.
  Then we can construct a~GBA accepting the complement of the language
  of~$\aut$ with~$\bigO(\tight(n+1)^k) = \bigO(n^k(0.76n)^{nk})$ states and~$k$
  colours.
\end{corollary}

\begin{proof}
\!\!$\bigO(\tight(n\!+\!1)^k)\!=\!
\bigO(\!(n\!\cdot\!\tight(n)\!)^k)\!=\!
\bigO(\!(n (0.76n)^n)^k)\!=\!
\bigO(n^k (0.76n)^{nk})$
\qed
\end{proof}

To the best of our knowledge, the state complexity of our procedure is
better than the complexity of other approaches (even if we require the output
to be a~BA and not a~GBA).
In particular, it is better than the complexity $\bigO(k\cdot 3^n \cdot
(2n+1)^{nk})$ of~\cite{KV05} and also better than the complexity of a~procedure
that would first transform the input Rabin automaton into a~BA with $m = nk$ states
and run the optimal BA complementation
with complexity~$\bigO(m(0.76m)^m) = \bigO(nk(0.76nk)^{nk})$~\cite{Schewe09},
\mbox{as shown by the following lemmas.}

\begin{lemma}\label{lem:complexity-rabin-kv}
$\bigO(n^k(0.76n)^{nk}) \subset \bigO(k\cdot 3^n \cdot (2n+1)^{nk})$
\end{lemma}

\begin{proof}
$n^k(0.76n)^{nk} = (\sqrt[n]{n}\cdot 0.76n)^{nk}$.  The global maximum of the
function $\sqrt[n]{n}$ is less than 1.5, so $(\sqrt[n]{n}\cdot 0.76n)^{nk} <
(1.14n)^{nk} < (2n+1)^{nk}$ for $n \geq 1$.
\qed
\end{proof}

\begin{lemma}
$\bigO(n^k(0.76n)^{nk}) \subset \bigO(nk(0.76nk)^{nk})$
\end{lemma}

\begin{proof}
Similar reasoning as in the proof of \cref{lem:complexity-rabin-kv}.
\qed
\end{proof}

%

\subsection{Parity Automata}

Since the parity condition is a~special case of the Rabin
condition~\cite{GradelTW01}, we can easily give an upper bound on the
complementation of parity automata.


\begin{lemma}
  For a~parity automaton~$\aut$ with index~$k$,
  there is a~GBA for the complement of~$\langof \aut$ with
  \!$\frac k 2$\! colours and $\bigOof{\tight(n\!+\!1)^{\frac{k}{2}}} \!=\!
  \bigOof{n^{\frac k 2} (0.76n)^{\frac{nk}{2}}}$ states. 
\end{lemma}

\begin{proof}

The min-odd parity acceptance condition is of the form $\acccond = 
\accfinof{\tacc 0} \land (\accinfof{\tacc 1} \lor (\accfinof{\tacc 2} \land (\accinfof{\tacc 3} \lor (\accfinof{\tacc 4} \land \ldots))))$.
If we transform the acceptance condition into the DNF, we obtain 
$\acccond' = (\accfinof{\tacc 0} \land \accinfof{\tacc 1}) \lor (\accfinof{{\tacc 0}+{\tacc 2}} \land \accinfof{{\tacc 1}+{\tacc 3}}) \lor 
(\accfinof{{\tacc 0}+{\tacc 2}+{\tacc 4}} \land \accinfof{{\tacc 1}+{\tacc 3}+{\tacc 5}}) \lor \ldots$
which is a~Rabin acceptance condition with $\frac{k}{2}$ Rabin pairs. 
Note that we can use a~new colour for each union of colours and we obtain the same number of colours 
as in $\acccond$. 
According to \cref{theorem:rabin}, the parity automaton $\aut$ can be complemented into a~GBA 
with $\bigOof{\tight(n+1)^{\frac{k}{2}}}$ states. 
\qed
\end{proof}

\noindent
We note that the complexity obtained by our general procedure is worse than the
best one we are aware of, which is $2^{\bigO(n \log n)}$~\cite{CaiZ11}.

\vspace{-0.0mm}
\subsection{Generalized Rabin Automata}
\vspace{-0.0mm}

Recall that the generalized Rabin condition is of the form 
$\accfinof{\tacc 0} \land \bigwedge_{j = 1}^n\accinfof{\taccj}$.
We can now easily combine the procedure for (standard) Rabin automata from the
previous section and the procedure for $\accinf$-TELA from \cref{subsec:infela}
to construct the subprocedure $\instbinf$ for $\bigwedge_{j =
1}^n\accinfof{\taccj}$.
The set of macrostates will be 
\begin{equation*}
\minbinf = 2^\states \cup
(\cT\times 2^\states \times \{ 0, 2, \dots, 2n - 2 \} \times \levelmodels) \cup
  (\cT\times \{ 0, 2, \dots, 2n - 2 \} \times \levelmodels )
\end{equation*}
Details are given in \cref{app:grabin}.
Similarly to \cref{subsec:infela,sec:rabin}, one can then obtain the following
bound on the size of the complement.

\begin{lemma}
  \label{lem:generalized_rabin}
  Let $\aut$ be a generalized Rabin automaton with one generalized Rabin pair
  with $\ell$ $\accinf$s. Then, there exists a BA accepting the complement
  of~$\aut$ with $\bigO(\ell^{n}\tight(n+1)) = \bigO(n\ell^n (0.76n)^n)$ states.
\end{lemma}

\begin{theorem}\label{thm:generalized_rabin_compl}
  Let $\aut$ be a generalized Rabin automaton with $k$ generalized Rabin pairs, each 
  with at most $\ell$ $\accinf$s. Then, there exists a~GBA with~$k$ colours and 
  $\bigO(\ell^{nk}\tight(n+1)^k) = \bigO(n^k(0.76 \ell n)^{nk})$ states
  accepting $\Sigma^\omega \setminus \langof \aut$.
\end{theorem}

\noindent
There is not much work on the complementation of generalized Rabin automata or
general TELAs (we are only aware of the upper bound~$2^{2^{\bigO(n)}}$
from~\cite{SafraV89})).
One could approach the complementation by translation of the generalized Rabin automaton into a~GBA
using the technique from~\cite{JohnJBK22}.
The technique first performs $\accfin$-removal, i.e., it makes $k$ copies
of~$\aut$, each with the corresponding $\accfin$-transitions removed, obtaining
a~GBA with $n(k+1)$ states and~$\ell$ colours (one can share colours across the
independent copies).
After that, we could use our GBA complementation algorithm from
\cref{sec:inftela}, which would give us a~BA with 
$\bigO(n(k+1)(0.76\ell n(k+1))^{n(k+1)})$ states, which is worse.

\begin{lemma}\label{lem:}
$\bigO(n^k(0.76 \ell n)^{nk}) \subset \bigO(n(k+1)(0.76\ell n(k+1))^{n(k+1)})$
\end{lemma}

\begin{proof}[Idea]
Let us observe the behaviour of the fraction with a~simplified right-hand side:
$\frac{nk(0.76\ell nk)^{nk}}{n^k(0.76 \ell n)^{nk}} =
\frac{nk^{nk+1}}{n^k}$.
There are two options:
\begin{enumerate}[(i)]
  \item  $n \geq k$: in this case, $k^{nk} \gg n^k$ and the theorem holds.
  \item  $k \ge n$: in this case, $k^k \gg n^k$ and the theorem holds.
\qed
\end{enumerate}
\end{proof}

\vspace{-0.0mm}
\subsection{General TELAs}
\label{general_ela}
\vspace{-0.0mm}

For complementation of general TELAs, we use the fact that any TELA can be
converted into a~generalized Rabin automaton with the same structure by
modifying the acceptance condition into the DNF form (and not touching the
structure of the automaton).
For a~TELA with~$k$ colours, the DNF will have at most~$2^k$ clauses (i.e.,
generalized Rabin pairs), each one with at most~$k$ literals.

\begin{theorem}\label{thm:}
Let $\aut$ be a TELA with $k$ colours.
Then, there exists a~GBA with~$2^k$ colours and 
$\bigO(k^{n2^k}\tight(n+1)^{2^k}) = \bigO(n^{2^k}(0.76 nk)^{n2^k})$ states
accepting $\Sigma^\omega \setminus \langof \aut$.
\end{theorem}

\begin{proof}
By substituting to \cref{thm:generalized_rabin_compl}.
\qed
\end{proof}

\ifHP
{
\vspace{-0.0mm}
\section{Related Work}\label{sec:label}
\vspace{-0.0mm}


Lower bounds for complementation of classes of $\omega$-automata using the full
automata technique were established in~\cite{Yan06} (improving the previous
$\Omega(n!)$ lower bound of Michel~\cite{michel1988complementation}).
The technique was later generalized to improve the lower bound of Rabin
automata complementation~\cite{rabin-lower}. 
Double exponential lower bound for complementation of general Emerson-Lei automata was
given in~\cite{SafraV89}.
\mbox{See the survey in~\cite{Boker18} for more details.}

Simultaneously to establishing the lower bound, there emerged algorithms for determinizing and complementing various classes of $\omega$-automata. 
The optimal determinization approach for GBAs introduced in \cite{ScheweV12} yields deterministic Rabin automaton of $\mathsf{ghist}_k(n)$ states 
and $2^n-1$ Rabin pairs, where $\mathsf{ghist}_k(n)$ converges against $(1.47 nk)^n$ for large $k$.
Rank-based complementation of GBAs was proposed in~\cite{KupfermanV05}. Furthermore, 
there are approaches for semideterminization-based complementation of
GBAs~\cite{seminator} with double exponential complexity.
Regarding other acceptance conditions, determinization of parity automata based on root history trees was proposed in~\cite{ScheweV14}. 
A rank-based complementation of Streett and Rabin automata was introduced in~\cite{KV05} and later 
improved by tree structures in~\cite{CaiZ11}. Tight determinization of Streett automata was presented in~\cite{streett-tight}.
Tight complementation technique for parity automata based on flattened nested history trees was then proposed in~\cite{ScheweV14a}.
A lot of effort has been put into complementation of B\"{u}chi automata leading to algorithms roughly divided into several groups:
Ramsey-based~\cite{breuers-improved-ramsey,buchi1962decision,sistla1987complementation},
rank-based~\cite{HavlenaL21,HavlenaLS22,HavlenaLS22b,vardi2007buchi,KupfermanV01,Schewe09}, 
determinization-based~\cite{safra1988complexity,piterman2006nondeterministic,LiTFVZ22}, slice-based~\cite{kahler2008complementation}, and others~\cite{AllredU18,HavlenaLLST23,li2018learning}.
There are specialized more efficient algorithms for subclasses of
BAs, such as inherently-weaks~\cite{MiyanoH84}, deterministic~\cite{Kurshan87},
semideterministic~\cite{BlahoudekHSST16},
elevator~\cite{HavlenaLS22,HavlenaLLST23}, or
unambiguous~\cite{li-unambigous,FengLTVZ23} BAs.
















\bibliographystyle{plainurl}
\bibliography{literature}

\begin{thebibliography}{10}

\bibitem{AllredU18}
Jo{\"{e}}l~D. Allred and Ulrich Ultes{-}Nitsche.
\newblock A simple and optimal complementation algorithm for {\buchi} automata.
\newblock In Anuj Dawar and Erich Gr{\"{a}}del, editors, {\em Proceedings of
  the 33rd Annual {ACM/IEEE} Symposium on Logic in Computer Science, {LICS}
  2018, Oxford, UK, July 09-12, 2018}, pages 46--55. {ACM}, 2018.
\newblock \href {https://doi.org/10.1145/3209108.3209138}
  {\path{doi:10.1145/3209108.3209138}}.

\bibitem{BlahoudekHSST16}
Frantisek Blahoudek, Matthias Heizmann, Sven Schewe, Jan Strejcek, and
  Ming{-}Hsien Tsai.
\newblock Complementing semi-deterministic {B{\"{u}}chi} automata.
\newblock In Marsha Chechik and Jean{-}Fran{\c{c}}ois Raskin, editors, {\em
  Tools and Algorithms for the Construction and Analysis of Systems - 22nd
  International Conference, {TACAS} 2016, Held as Part of the European Joint
  Conferences on Theory and Practice of Software, {ETAPS} 2016, Eindhoven, The
  Netherlands, April 2-8, 2016, Proceedings}, volume 9636 of {\em Lecture Notes
  in Computer Science}, pages 770--787. Springer, 2016.
\newblock \href {https://doi.org/10.1007/978-3-662-49674-9\_49}
  {\path{doi:10.1007/978-3-662-49674-9\_49}}.

\bibitem{seminator}
Franti\v{s}ek Blahoudek, Alexandre Duret-Lutz, and Jan Strej\v{c}ek.
\newblock {S}eminator~2 can complement generalized {B\"u}chi automata via
  improved semi-determinization.
\newblock In {\em Proceedings of the 32nd International Conference on
  Computer-Aided Verification (CAV'20)}, volume 12225 of {\em Lecture Notes in
  Computer Science}, pages 15--27. Springer, July 2020.
\newblock \href {https://doi.org/10.1007/978-3-030-53291-8_2}
  {\path{doi:10.1007/978-3-030-53291-8_2}}.

\bibitem{Boker18}
Udi Boker.
\newblock Why these automata types?
\newblock In {\em {LPAR-22.} 22nd International Conference on Logic for
  Programming, Artificial Intelligence and Reasoning, Awassa, Ethiopia, 16-21
  November 2018}, volume~57 of {\em EPiC Series in Computing}, pages 143--163.
  EasyChair, 2018.

\bibitem{breuers-improved-ramsey}
Stefan Breuers, Christof L{\"o}ding, and J{\"o}rg Olschewski.
\newblock Improved {R}amsey-based {B{\"u}chi} complementation.
\newblock In {\em Proc. of FOSSACS'12}, pages 150--164. Springer, 2012.

\bibitem{buchi1962decision}
J.~Richard B{\"u}chi.
\newblock On a decision method in restricted second order arithmetic.
\newblock In {\em Proc. of International Congress on Logic, Method, and
  Philosophy of Science 1960}. Stanford Univ. Press, Stanford, 1962.

\bibitem{CaiZ11}
Yang Cai and Ting Zhang.
\newblock Tight upper bounds for {Streett} and parity complementation.
\newblock In Marc Bezem, editor, {\em Computer Science Logic, 25th
  International Workshop / 20th Annual Conference of the EACSL, {CSL} 2011,
  September 12-15, 2011, Bergen, Norway, Proceedings}, volume~12 of {\em
  LIPIcs}, pages 112--128. Schloss Dagstuhl - Leibniz-Zentrum f{\"{u}}r
  Informatik, 2011.
\newblock \href {https://doi.org/10.4230/LIPIcs.CSL.2011.112}
  {\path{doi:10.4230/LIPIcs.CSL.2011.112}}.

\bibitem{rabin-lower}
Yang Cai, Ting Zhang, and Haifeng Luo.
\newblock An improved lower bound for the complementation of rabin automata.
\newblock In {\em Proceedings of the 2009 24th Annual IEEE Symposium on Logic
  In Computer Science}, LICS '09, page 167–176, USA, 2009. IEEE Computer
  Society.
\newblock \href {https://doi.org/10.1109/LICS.2009.13}
  {\path{doi:10.1109/LICS.2009.13}}.

\bibitem{ChenHL19}
Yu{-}Fang Chen, Vojtech Havlena, and Ondrej Leng{\'{a}}l.
\newblock Simulations in rank-based {B{\"{u}}chi} automata complementation.
\newblock In Anthony~Widjaja Lin, editor, {\em Programming Languages and
  Systems - 17th Asian Symposium, {APLAS} 2019, Nusa Dua, Bali, Indonesia,
  December 1-4, 2019, Proceedings}, volume 11893 of {\em Lecture Notes in
  Computer Science}, pages 447--467. Springer, 2019.
\newblock \href {https://doi.org/10.1007/978-3-030-34175-6\_23}
  {\path{doi:10.1007/978-3-030-34175-6\_23}}.

\bibitem{ClarksonFKMRS14}
Michael~R. Clarkson, Bernd Finkbeiner, Masoud Koleini, Kristopher~K. Micinski,
  Markus~N. Rabe, and C{\'{e}}sar S{\'{a}}nchez.
\newblock Temporal logics for hyperproperties.
\newblock In Mart{\'{\i}}n Abadi and Steve Kremer, editors, {\em Principles of
  Security and Trust - Third International Conference, {POST} 2014, Held as
  Part of the European Joint Conferences on Theory and Practice of Software,
  {ETAPS} 2014, Grenoble, France, April 5-13, 2014, Proceedings}, volume 8414
  of {\em Lecture Notes in Computer Science}, pages 265--284. Springer, 2014.
\newblock \href {https://doi.org/10.1007/978-3-642-54792-8\_15}
  {\path{doi:10.1007/978-3-642-54792-8\_15}}.

\bibitem{EmersonL87}
E.~Allen Emerson and Chin{-}Laung Lei.
\newblock Modalities for model checking: Branching time logic strikes back.
\newblock {\em Sci. Comput. Program.}, 8(3):275--306, 1987.
\newblock \href {https://doi.org/10.1016/0167-6423(87)90036-0}
  {\path{doi:10.1016/0167-6423(87)90036-0}}.

\bibitem{FengLTVZ23}
Weizhi Feng, Yong Li, Andrea Turrini, Moshe~Y. Vardi, and Lijun Zhang.
\newblock On the power of finite ambiguity in \buchi complementation.
\newblock {\em Inf. Comput.}, 292:105032, 2023.
\newblock URL: \url{https://doi.org/10.1016/j.ic.2023.105032}, \href
  {https://doi.org/10.1016/J.IC.2023.105032}
  {\path{doi:10.1016/J.IC.2023.105032}}.

\bibitem{FriedgutKV06}
Ehud Friedgut, Orna Kupferman, and Moshe~Y. Vardi.
\newblock \buchi complementation made tighter.
\newblock {\em Int. J. Found. Comput. Sci.}, 17(4):851--868, 2006.
\newblock \href {https://doi.org/10.1142/S0129054106004145}
  {\path{doi:10.1142/S0129054106004145}}.

\bibitem{GradelTW01}
Erich Gr{\"{a}}del, Wolfgang Thomas, and Thomas Wilke, editors.
\newblock {\em Automata, Logics, and Infinite Games: {A} Guide to Current
  Research [outcome of a Dagstuhl seminar, February 2001]}, volume 2500 of {\em
  Lecture Notes in Computer Science}. Springer, 2002.
\newblock \href {https://doi.org/10.1007/3-540-36387-4}
  {\path{doi:10.1007/3-540-36387-4}}.

\bibitem{HavlenaL21}
Vojtech Havlena and Ondrej Leng{\'{a}}l.
\newblock Reducing (to) the ranks: Efficient rank-based \buchi automata
  complementation.
\newblock In Serge Haddad and Daniele Varacca, editors, {\em 32nd International
  Conference on Concurrency Theory, {CONCUR} 2021, August 24-27, 2021, Virtual
  Conference}, volume 203 of {\em LIPIcs}, pages 2:1--2:19. Schloss Dagstuhl -
  Leibniz-Zentrum f{\"{u}}r Informatik, 2021.
\newblock URL: \url{https://doi.org/10.4230/LIPIcs.CONCUR.2021.2}, \href
  {https://doi.org/10.4230/LIPICS.CONCUR.2021.2}
  {\path{doi:10.4230/LIPICS.CONCUR.2021.2}}.

\bibitem{HavlenaLLST23}
Vojt\v{e}ch Havlena, Ond\v{r}ej Leng{\'{a}}l, Yong Li, Barbora
  \v{S}mahl{\'{\i}}kov{\'{a}}, and Andrea Turrini.
\newblock Modular mix-and-match complementation of {B{\"{u}}chi} automata.
\newblock In {\em Tools and Algorithms for the Construction and Analysis of
  Systems - 28th International Conference, {TACAS} 2023, Held as Part of the
  European Joint Conferences on Theory and Practice of Software, {ETAPS} 2023,
  Paris, France}, Lecture Notes in Computer Science. Springer, 2023.

\bibitem{HavlenaLS22b}
Vojt\v{e}ch Havlena, Ond\v{r}ej Leng{\'{a}}l, and Barbora
  \v{S}mahl{\'{\i}}kov{\'{a}}.
\newblock Complementing b{\"{u}}chi automata with ranker.
\newblock In Sharon Shoham and Yakir Vizel, editors, {\em Computer Aided
  Verification - 34th International Conference, {CAV} 2022, Haifa, Israel,
  August 7-10, 2022, Proceedings, Part {II}}, volume 13372 of {\em Lecture
  Notes in Computer Science}, pages 188--201. Springer, 2022.
\newblock \href {https://doi.org/10.1007/978-3-031-13188-2\_10}
  {\path{doi:10.1007/978-3-031-13188-2\_10}}.

\bibitem{HavlenaLS22}
Vojt\v{e}ch Havlena, Ond\v{r}ej Leng{\'{a}}l, and Barbora
  \v{S}mahl{\'{\i}}kov{\'{a}}.
\newblock Sky is not the limit: Tighter rank bounds for elevator automata in
  {B{\"{u}}chi} automata complementation.
\newblock In Dana Fisman and Grigore Rosu, editors, {\em Tools and Algorithms
  for the Construction and Analysis of Systems - 28th International Conference,
  {TACAS} 2022, Held as Part of the European Joint Conferences on Theory and
  Practice of Software, {ETAPS} 2022, Munich, Germany, April 2-7, 2022,
  Proceedings, Part {II}}, volume 13244 of {\em Lecture Notes in Computer
  Science}, pages 118--136. Springer, 2022.
\newblock \href {https://doi.org/10.1007/978-3-030-99527-0\_7}
  {\path{doi:10.1007/978-3-030-99527-0\_7}}.

\bibitem{HieronymiMOSSS24}
Philipp Hieronymi, Dun Ma, Reed Oei, Luke Schaeffer, Christian Schulz, and
  Jeffrey~O. Shallit.
\newblock Decidability for {Sturmian} words.
\newblock {\em Log. Methods Comput. Sci.}, 20(3), 2024.
\newblock URL: \url{https://doi.org/10.46298/lmcs-20(3:12)2024}, \href
  {https://doi.org/10.46298/LMCS-20(3:12)2024}
  {\path{doi:10.46298/LMCS-20(3:12)2024}}.

\bibitem{JohnJBK22}
Tobias John, Simon Jantsch, Christel Baier, and Sascha Kl{\"{u}}ppelholz.
\newblock From {Emerson-Lei} automata to deterministic, limit-deterministic or
  good-for-{MDP} automata.
\newblock {\em Innov. Syst. Softw. Eng.}, 18(3):385--403, 2022.
\newblock \href {https://doi.org/10.1007/s11334-022-00445-7}
  {\path{doi:10.1007/s11334-022-00445-7}}.

\bibitem{kahler2008complementation}
Detlef K{\"a}hler and Thomas Wilke.
\newblock Complementation, disambiguation, and determinization of {B{\"u}chi}
  automata unified.
\newblock In {\em Proc. of ICALP'08}, pages 724--735. Springer, 2008.

\bibitem{KarmarkarC09}
Hrishikesh Karmarkar and Supratik Chakraborty.
\newblock On minimal odd rankings for {B{\"{u}}chi} complementation.
\newblock In Zhiming Liu and Anders~P. Ravn, editors, {\em Automated Technology
  for Verification and Analysis, 7th International Symposium, {ATVA} 2009,
  Macao, China, October 14-16, 2009. Proceedings}, volume 5799 of {\em Lecture
  Notes in Computer Science}, pages 228--243. Springer, 2009.
\newblock \href {https://doi.org/10.1007/978-3-642-04761-9\_18}
  {\path{doi:10.1007/978-3-642-04761-9\_18}}.

\bibitem{KestenP95}
Yonit Kesten and Amir Pnueli.
\newblock A complete proof systems for {QPTL}.
\newblock In {\em Proceedings, 10th Annual {IEEE} Symposium on Logic in
  Computer Science, San Diego, California, USA, June 26-29, 1995}, pages 2--12.
  {IEEE} Computer Society, 1995.
\newblock \href {https://doi.org/10.1109/LICS.1995.523239}
  {\path{doi:10.1109/LICS.1995.523239}}.

\bibitem{KV05}
Orna Kupferman and Moshe Vardi.
\newblock Complementation constructions for nondeterministic automata on
  infinite words.
\newblock volume 3440, pages 206--221, 04 2005.
\newblock \href {https://doi.org/10.1007/978-3-540-31980-1_14}
  {\path{doi:10.1007/978-3-540-31980-1_14}}.

\bibitem{KupfermanV01}
Orna Kupferman and Moshe~Y. Vardi.
\newblock Weak alternating automata are not that weak.
\newblock {\em {ACM} Trans. Comput. Log.}, 2(3):408--429, 2001.
\newblock \href {https://doi.org/10.1145/377978.377993}
  {\path{doi:10.1145/377978.377993}}.

\bibitem{KupfermanV05}
Orna Kupferman and Moshe~Y. Vardi.
\newblock From complementation to certification.
\newblock {\em Theor. Comput. Sci.}, 345(1):83--100, 2005.
\newblock \href {https://doi.org/10.1016/j.tcs.2005.07.021}
  {\path{doi:10.1016/j.tcs.2005.07.021}}.

\bibitem{Kurshan87}
Robert~P. Kurshan.
\newblock Complementing deterministic {B{\"{u}}chi} automata in polynomial
  time.
\newblock {\em J. Comput. Syst. Sci.}, 35(1):59--71, 1987.
\newblock \href {https://doi.org/10.1016/0022-0000(87)90036-5}
  {\path{doi:10.1016/0022-0000(87)90036-5}}.

\bibitem{LiTFVZ22}
Yong Li, Andrea Turrini, Weizhi Feng, Moshe~Y. Vardi, and Lijun Zhang.
\newblock Divide-and-conquer determinization of \buchi automata based on {SCC}
  decomposition.
\newblock In Sharon Shoham and Yakir Vizel, editors, {\em Computer Aided
  Verification - 34th International Conference, {CAV} 2022, Haifa, Israel,
  August 7-10, 2022, Proceedings, Part {II}}, volume 13372 of {\em Lecture
  Notes in Computer Science}, pages 152--173. Springer, 2022.
\newblock \href {https://doi.org/10.1007/978-3-031-13188-2\_8}
  {\path{doi:10.1007/978-3-031-13188-2\_8}}.

\bibitem{li2018learning}
Yong Li, Andrea Turrini, Lijun Zhang, and Sven Schewe.
\newblock Learning to complement {B{\"u}chi} automata.
\newblock In {\em Proc. of VMCAI'18}, pages 313--335. Springer, 2018.

\bibitem{li-unambigous}
Yong Li, Moshe~Y. Vardi, and Lijun Zhang.
\newblock On the power of unambiguity in {B\"uchi} complementation.
\newblock In Jean-Francois Raskin and Davide Bresolin, editors, {\em {\rm
  Proceedings 11th International Symposium on} Games, Automata, Logics, and
  Formal Verification, {\rm Brussels, Belgium, September 21-22, 2020}}, volume
  326 of {\em Electronic Proceedings in Theoretical Computer Science}, pages
  182--198. Open Publishing Association, 2020.
\newblock \href {https://doi.org/10.4204/EPTCS.326.12}
  {\path{doi:10.4204/EPTCS.326.12}}.

\bibitem{michel1988complementation}
Max Michel.
\newblock Complementation is more difficult with automata on infinite words.
\newblock {\em CNET, Paris}, 15, 1988.

\bibitem{MiyanoH84}
Satoru Miyano and Takeshi Hayashi.
\newblock Alternating finite automata on omega-words.
\newblock In Bruno Courcelle, editor, {\em CAAP'84, 9th Colloquium on Trees in
  Algebra and Programming, Bordeaux, France, March 5-7, 1984, Proceedings},
  pages 195--210. Cambridge University Press, 1984.

\bibitem{piterman2006nondeterministic}
Nir Piterman.
\newblock From nondeterministic {B{\"u}chi} and {Streett} automata to
  deterministic parity automata.
\newblock In {\em Proc. of LICS'06}, pages 255--264. IEEE, 2006.

\bibitem{safra1988complexity}
Shmuel Safra.
\newblock On the complexity of $\omega$-automata.
\newblock In {\em Proc. of FOCS'88}, pages 319--327. IEEE, 1988.

\bibitem{SafraV89}
Shmuel Safra and Moshe~Y. Vardi.
\newblock On $\omega$-automata and temporal logic (preliminary report).
\newblock In David~S. Johnson, editor, {\em Proceedings of the 21st Annual
  {ACM} Symposium on Theory of Computing, May 14-17, 1989, Seattle, Washington,
  {USA}}, pages 127--137. {ACM}, 1989.
\newblock \href {https://doi.org/10.1145/73007.73019}
  {\path{doi:10.1145/73007.73019}}.

\bibitem{Schewe09}
Sven Schewe.
\newblock {B{\"{u}}chi} complementation made tight.
\newblock In Susanne Albers and Jean{-}Yves Marion, editors, {\em 26th
  International Symposium on Theoretical Aspects of Computer Science, {STACS}
  2009, February 26-28, 2009, Freiburg, Germany, Proceedings}, volume~3 of {\em
  LIPIcs}, pages 661--672. Schloss Dagstuhl - Leibniz-Zentrum fuer Informatik,
  Germany, 2009.
\newblock \href {https://doi.org/10.4230/LIPIcs.STACS.2009.1854}
  {\path{doi:10.4230/LIPIcs.STACS.2009.1854}}.

\bibitem{ScheweV12}
Sven Schewe and Thomas Varghese.
\newblock Tight bounds for the determinisation and complementation of
  generalised {B{\"{u}}chi} automata.
\newblock In Supratik Chakraborty and Madhavan Mukund, editors, {\em Automated
  Technology for Verification and Analysis - 10th International Symposium,
  {ATVA} 2012, Thiruvananthapuram, India, October 3-6, 2012. Proceedings},
  volume 7561 of {\em Lecture Notes in Computer Science}, pages 42--56.
  Springer, 2012.
\newblock \href {https://doi.org/10.1007/978-3-642-33386-6\_5}
  {\path{doi:10.1007/978-3-642-33386-6\_5}}.

\bibitem{ScheweV14}
Sven Schewe and Thomas Varghese.
\newblock Determinising parity automata.
\newblock In Erzs{\'{e}}bet Csuhaj{-}Varj{\'{u}}, Martin Dietzfelbinger, and
  Zolt{\'{a}}n {\'{E}}sik, editors, {\em Mathematical Foundations of Computer
  Science 2014 - 39th International Symposium, {MFCS} 2014, Budapest, Hungary,
  August 25-29, 2014. Proceedings, Part {I}}, volume 8634 of {\em Lecture Notes
  in Computer Science}, pages 486--498. Springer, 2014.
\newblock \href {https://doi.org/10.1007/978-3-662-44522-8\_41}
  {\path{doi:10.1007/978-3-662-44522-8\_41}}.

\bibitem{ScheweV14a}
Sven Schewe and Thomas Varghese.
\newblock Tight bounds for complementing parity automata.
\newblock In Erzs{\'{e}}bet Csuhaj{-}Varj{\'{u}}, Martin Dietzfelbinger, and
  Zolt{\'{a}}n {\'{E}}sik, editors, {\em Mathematical Foundations of Computer
  Science 2014 - 39th International Symposium, {MFCS} 2014, Budapest, Hungary,
  August 25-29, 2014. Proceedings, Part {I}}, volume 8634 of {\em Lecture Notes
  in Computer Science}, pages 499--510. Springer, 2014.
\newblock \href {https://doi.org/10.1007/978-3-662-44522-8\_42}
  {\path{doi:10.1007/978-3-662-44522-8\_42}}.

\bibitem{sistla1987complementation}
A.~Prasad Sistla, Moshe~Y. Vardi, and Pierre Wolper.
\newblock {The Complementation Problem for {B{\"u}chi} Automata with
  Applications to Temporal Logic}.
\newblock {\em Theoretical Computer Science}, 49(2-3):217--237, 1987.

\bibitem{streett-tight}
Cong Tian, Wensheng Wang, and Zhenhua Duan.
\newblock Making streett determinization tight.
\newblock In {\em Proceedings of the 35th Annual ACM/IEEE Symposium on Logic in
  Computer Science}, LICS '20, page 859–872, New York, NY, USA, 2020.
  Association for Computing Machinery.
\newblock \href {https://doi.org/10.1145/3373718.3394757}
  {\path{doi:10.1145/3373718.3394757}}.

\bibitem{vardi2007buchi}
Moshe~Y. Vardi.
\newblock The {B{\"u}chi} complementation saga.
\newblock In {\em Proc. of STACS'07}, pages 12--22. Springer, 2007.

\bibitem{Yan06}
Qiqi Yan.
\newblock Lower bounds for complementation of $\omega$-automata via the full
  automata technique.
\newblock In Michele Bugliesi, Bart Preneel, Vladimiro Sassone, and Ingo
  Wegener, editors, {\em Automata, Languages and Programming}, pages 589--600,
  Berlin, Heidelberg, 2006. Springer Berlin Heidelberg.

\end{thebibliography}

\newpage
\appendix
\newcommand{\figRelRunDAG}[0]{
\begin{figure}[t]
  \begin{subfigure}[b]{0.32\linewidth}
  \begin{center}
  \begin{tikzpicture}[>=stealth',shorten >=0pt,auto,node distance=1.0cm,
                    scale=0.8,transform shape,initial text={}]
  \tikzstyle{every state}=[inner sep=3pt,minimum size=5pt,rectangle,rounded corners=1mm]
  \tikzstyle{empty}=[]

  \node[state] (q0) {$q,0$};

  \node[state,below of=q0] (q1) {$q,1$};
  \node[state,right of=q1] (r1) {$r,1$};
  \node[right of=r1] (t1) {};
  \coordinate[above right of=t4,xshift=-3mm] (t4_ar);

  \node[state,below of=q1] (q2) {$q,2$};
  \node[below of=r1]       (r2) {};
  \node[state,right of=r2] (t2) {$t,2$};
  \node[state,right of=t2] (s2) {$s,2$};

  \node[below of=r2]      (r3) {};
  \node[state,left of=r3] (q3) {$q,3$};
  \node[below of=s2]      (s3) {};

  \node[below of=q3] (q4) {$\vdots$};
  \coordinate[below left of=q4,xshift=3mm,yshift=3mm] (q4_bl);
  \node[right of=q4] (r4) {$\ddots$};
  \node[right of=r4] (t4) {};
  \coordinate[below right of=t4,xshift=-3mm,yshift=3mm] (t4_br);

  \draw[->] (q0) edge node[above right,xshift=-1mm,yshift=-1mm] {\tacc 0} (r1)
    (q0) edge (q1)
    (r1) edge (s2)
    (r1) edge node[below left,xshift=1mm,yshift=1mm] {\tacc 2} (t2)
    (q1) edge (q2);

  \draw[->] (t2) edge (q3)
    (q2) edge (q3);

  \draw[->] (q3) edge (q4)
    (q3) edge node[above right,xshift=-1mm,yshift=-1mm] {\tacc 0} (r4);

  \node[empty, node distance=8mm,below left of=q0,xshift=-3mm] (sym1) {$c$};
  \node[empty, below of=sym1] (sym2) {$a$};
  \node[empty, below of=sym2] (sym3) {$b$};
  \node[empty, below of=sym3,node distance=10mm] (symdots) {$\vdots$};
\end{tikzpicture}
  \caption{The run DAG $\dagw$}
  \end{center}
  \label{label}
  \end{subfigure}
  \begin{subfigure}[b]{0.29\linewidth}
  \begin{center}
    \begin{tikzpicture}[>=stealth',shorten >=0pt,auto,node distance=1.0cm,
                    scale=0.8,transform shape,initial text={}]
  \tikzstyle{every state}=[inner sep=3pt,minimum size=5pt,rectangle,rounded corners=1mm]
  \tikzstyle{empty}=[]

  \node[state] (q0) {$q,0$};

  \node[state,below of=q0] (q1) {$q,1$};
  \node[state,right of=q1] (r1) {$r,1$};
  \node[right of=r1] (t1) {};
  \coordinate[above right of=t4,xshift=-3mm] (t4_ar);

  \node[state,below of=q1] (q2) {$q,2$};
  \node[below of=r1]       (r2) {};
  \node[state,right of=r2] (s2) {$s,2$};

  \node[below of=r2]      (r3) {};
  \node[state,left of=r3] (q3) {$q,3$};
  \node[below of=s2]      (s3) {};

  \node[below of=q3] (q4) {$\vdots$};
  \coordinate[below left of=q4,xshift=3mm,yshift=3mm] (q4_bl);
  \node[right of=q4] (r4) {$\ddots$};
  \node[right of=r4] (t4) {};
  \coordinate[below right of=t4,xshift=-3mm,yshift=3mm] (t4_br);

  \draw[->] (q0) edge node[above right,xshift=-1mm,yshift=-1mm] {\tacc 0} (r1)
    (q0) edge (q1)
    (r1) edge (s2)
    (q1) edge (q2);

  \draw[->] (q2) edge (q3);

  \draw[->] (q3) edge (q4)
    (q3) edge node[above right,xshift=-1mm,yshift=-1mm] {\tacc 0} (r4);

  \node[empty, node distance=8mm,below left of=q0,xshift=-3mm] (sym1) {$c$};
  \node[empty, below of=sym1] (sym2) {$a$};
  \node[empty, below of=sym2] (sym3) {$b$};
  \node[empty, below of=sym3,node distance=10mm] (symdots) {$\vdots$};
\end{tikzpicture}
    \caption{The narrow run DAG $\gendag$}
  \end{center}
  \label{label}
  \end{subfigure}
  \begin{subfigure}[b]{0.36\linewidth}
    \begin{center}
      \begin{tikzpicture}[>=stealth',shorten >=0pt,auto,node distance=1.0cm,
                    scale=0.8,transform shape,initial text={}]
  \tikzstyle{every state}=[inner sep=3pt,minimum size=5pt,rectangle,rounded corners=1mm]
  \tikzstyle{empty}=[]

  \node[state] (q0) {$q,0$};

  \node[state,below of=q0] (q1) {$q,1$};
  \node[state,right of=q1] (r1) {$r,1$};
  \node[right of=r1] (t1) {};
  \coordinate[above right of=t4,xshift=-3mm] (t4_ar);

  \node[state,below of=q1] (q2) {$q,2$};
  \node[below of=r1]       (r2) {};
  \node[state,right of=r2] (t2) {$t,2$};
  \node[state,right of=t2] (s2) {$s,2$};

  \node[below of=r2]      (r3) {};
  \node[state,left of=r3] (q3) {$q,3$};
  \node[below of=s2]      (s3) {};

  \node[below of=q3] (q4) {$\vdots$};
  \coordinate[below left of=q4,xshift=3mm,yshift=3mm] (q4_bl);
  \node[right of=q4] (r4) {$\ddots$};
  \node[right of=r4] (t4) {};
  \coordinate[below right of=t4,xshift=-3mm,yshift=3mm] (t4_br);

  \draw[->] (q0) edge node[above right,xshift=-1mm,yshift=-1mm] {\tacc 0} (r1)
    (q0) edge (q1)
    (r1) edge (s2)
    (q1) edge (q2);

  \draw[->] (t2) edge (q3)
    (q2) edge (q3);

  \draw[->] (q3) edge (q4)
    (q3) edge node[above right,xshift=-1mm,yshift=-1mm] {\tacc 0} (r4);

  \node[empty, node distance=8mm,below left of=q0,xshift=-3mm] (sym1) {$c$};
  \node[empty, below of=sym1] (sym2) {$a$};
  \node[empty, below of=sym2] (sym3) {$b$};
  \node[empty, below of=sym3,node distance=10mm] (symdots) {$\vdots$};
\end{tikzpicture}
    \end{center}
    \caption{The $\I$-narrow run DAG $\gendag{}^\prime$ for $\I = \{ 3k \mid k \in \omega \}$. }
    \label{label}
    \end{subfigure}
  \caption{Consider the TELA $\autex$ from \cref{fig:ex_inf_inf_aut} with the
  acceptance condition $\accinfof{\acccondof 0} \land \accfinof{\acccondof 2}$, the
  transition function $\delta$ and the transition function $\Delta$ omitting 
    transitions labelled by $\acccondof 2$ (i.e., the single transition $r \ltr a t$).
    Then, for a word $w = (cab)^\omega$, we show a particular run DAG in each subfigure.
    \ol{}
   }
  \label{fig:dags}
\end{figure}
}

\vspace{-0.0mm}
\section{Proofs for \cref{sec:fin-modular}}\label{sec:label}
\vspace{-0.0mm}

\begin{figure}[t]
  \begin{subfigure}[b]{0.32\linewidth}
  \begin{center}
  \begin{tikzpicture}[>=stealth',shorten >=0pt,auto,node distance=1.0cm,
                    scale=0.8,transform shape,initial text={}]
  \tikzstyle{every state}=[inner sep=3pt,minimum size=5pt,rectangle,rounded corners=1mm]
  \tikzstyle{empty}=[]

  \node[state] (q0) {$q,0$};

  \node[state,below of=q0] (q1) {$q,1$};
  \node[state,right of=q1] (r1) {$r,1$};
  \node[right of=r1] (t1) {};
  \coordinate[above right of=t4,xshift=-3mm] (t4_ar);

  \node[state,below of=q1] (q2) {$q,2$};
  \node[below of=r1]       (r2) {};
  \node[state,right of=r2] (t2) {$t,2$};
  \node[state,right of=t2] (s2) {$s,2$};

  \node[below of=r2]      (r3) {};
  \node[state,left of=r3] (q3) {$q,3$};
  \node[below of=s2]      (s3) {};

  \node[below of=q3] (q4) {$\vdots$};
  \coordinate[below left of=q4,xshift=3mm,yshift=3mm] (q4_bl);
  \node[right of=q4] (r4) {$\ddots$};
  \node[right of=r4] (t4) {};
  \coordinate[below right of=t4,xshift=-3mm,yshift=3mm] (t4_br);

  \draw[->] (q0) edge node[above right,xshift=-1mm,yshift=-1mm] {\tacc 0} (r1)
    (q0) edge (q1)
    (r1) edge (s2)
    (r1) edge node[below left,xshift=1mm,yshift=1mm] {\tacc 2} (t2)
    (q1) edge (q2);

  \draw[->] (t2) edge (q3)
    (q2) edge (q3);

  \draw[->] (q3) edge (q4)
    (q3) edge node[above right,xshift=-1mm,yshift=-1mm] {\tacc 0} (r4);

  \node[empty, node distance=8mm,below left of=q0,xshift=-3mm] (sym1) {$c$};
  \node[empty, below of=sym1] (sym2) {$a$};
  \node[empty, below of=sym2] (sym3) {$b$};
  \node[empty, below of=sym3,node distance=10mm] (symdots) {$\vdots$};
\end{tikzpicture}
  \caption{The run DAG $\dagw$}
  \end{center}
  \label{label}
  \end{subfigure}
  \begin{subfigure}[b]{0.29\linewidth}
  \begin{center}
    \begin{tikzpicture}[>=stealth',shorten >=0pt,auto,node distance=1.0cm,
                    scale=0.8,transform shape,initial text={}]
  \tikzstyle{every state}=[inner sep=3pt,minimum size=5pt,rectangle,rounded corners=1mm]
  \tikzstyle{empty}=[]

  \node[state] (q0) {$q,0$};

  \node[state,below of=q0] (q1) {$q,1$};
  \node[state,right of=q1] (r1) {$r,1$};
  \node[right of=r1] (t1) {};
  \coordinate[above right of=t4,xshift=-3mm] (t4_ar);

  \node[state,below of=q1] (q2) {$q,2$};
  \node[below of=r1]       (r2) {};
  \node[state,right of=r2] (s2) {$s,2$};

  \node[below of=r2]      (r3) {};
  \node[state,left of=r3] (q3) {$q,3$};
  \node[below of=s2]      (s3) {};

  \node[below of=q3] (q4) {$\vdots$};
  \coordinate[below left of=q4,xshift=3mm,yshift=3mm] (q4_bl);
  \node[right of=q4] (r4) {$\ddots$};
  \node[right of=r4] (t4) {};
  \coordinate[below right of=t4,xshift=-3mm,yshift=3mm] (t4_br);

  \draw[->] (q0) edge node[above right,xshift=-1mm,yshift=-1mm] {\tacc 0} (r1)
    (q0) edge (q1)
    (r1) edge (s2)
    (q1) edge (q2);

  \draw[->] (q2) edge (q3);

  \draw[->] (q3) edge (q4)
    (q3) edge node[above right,xshift=-1mm,yshift=-1mm] {\tacc 0} (r4);

  \node[empty, node distance=8mm,below left of=q0,xshift=-3mm] (sym1) {$c$};
  \node[empty, below of=sym1] (sym2) {$a$};
  \node[empty, below of=sym2] (sym3) {$b$};
  \node[empty, below of=sym3,node distance=10mm] (symdots) {$\vdots$};
\end{tikzpicture}
    \caption{The narrow run DAG $\gendag$}
  \end{center}
  \label{label}
  \end{subfigure}
  \begin{subfigure}[b]{0.36\linewidth}
    \begin{center}
      \begin{tikzpicture}[>=stealth',shorten >=0pt,auto,node distance=1.0cm,
                    scale=0.8,transform shape,initial text={}]
  \tikzstyle{every state}=[inner sep=3pt,minimum size=5pt,rectangle,rounded corners=1mm]
  \tikzstyle{empty}=[]

  \node[state] (q0) {$q,0$};

  \node[state,below of=q0] (q1) {$q,1$};
  \node[state,right of=q1] (r1) {$r,1$};
  \node[right of=r1] (t1) {};
  \coordinate[above right of=t4,xshift=-3mm] (t4_ar);

  \node[state,below of=q1] (q2) {$q,2$};
  \node[below of=r1]       (r2) {};
  \node[state,right of=r2] (t2) {$t,2$};
  \node[state,right of=t2] (s2) {$s,2$};

  \node[below of=r2]      (r3) {};
  \node[state,left of=r3] (q3) {$q,3$};
  \node[below of=s2]      (s3) {};

  \node[below of=q3] (q4) {$\vdots$};
  \coordinate[below left of=q4,xshift=3mm,yshift=3mm] (q4_bl);
  \node[right of=q4] (r4) {$\ddots$};
  \node[right of=r4] (t4) {};
  \coordinate[below right of=t4,xshift=-3mm,yshift=3mm] (t4_br);

  \draw[->] (q0) edge node[above right,xshift=-1mm,yshift=-1mm] {\tacc 0} (r1)
    (q0) edge (q1)
    (r1) edge (s2)
    (q1) edge (q2);

  \draw[->] (t2) edge (q3)
    (q2) edge (q3);

  \draw[->] (q3) edge (q4)
    (q3) edge node[above right,xshift=-1mm,yshift=-1mm] {\tacc 0} (r4);

  \node[empty, node distance=8mm,below left of=q0,xshift=-3mm] (sym1) {$c$};
  \node[empty, below of=sym1] (sym2) {$a$};
  \node[empty, below of=sym2] (sym3) {$b$};
  \node[empty, below of=sym3,node distance=10mm] (symdots) {$\vdots$};
\end{tikzpicture}
    \end{center}
    \caption{The $\I$-narrow run DAG $\gendag{}^\prime$ for $\I = \{ 3k \mid k \in \omega \}$. }
    \label{label}
    \end{subfigure}
  \caption{Consider the TELA $\autex$ from \cref{fig:ex_inf_inf_aut} with the
  acceptance condition $\accinfof{\acccondof 0} \land \accfinof{\acccondof 2}$, the
  transition function $\delta$ and the transition function $\Delta$ omitting 
    transitions labelled by $\acccondof 2$ (i.e., the single transition $r \ltr a t$).
    Then, for a word $w = (cab)^\omega$, we show a particular run DAG in each subfigure.
    \ol{}
   }
  \label{fig:dags}
\end{figure}

Let $\I \subseteq \omega$ be a set of indices.
An~RRDAG $\gendag = (S_0, S_1, \dots)$ is $\I$-narrow if $S_0 = I$ and for all
$i \in \omega$ it holds that
\begin{enumerate}[(i)]
  \item $\Delta(S_i, \wordof{i}) = S_{i+1}$ if $i+1 \notin \I$ and
  \item $\Delta(S_i, \wordof{i}) \subseteq S_{i+1}$ otherwise.
\end{enumerate}
An RRDAG is called \emph{narrow} if it is $\emptyset$-narrow (note that for
every word, there is exactly one narrow RRDAG).
Let $\dagg_1 = (S_0, S_1, \dots)$ and $\dagg_2 = (S_0', S_1', \dots)$ be two RRDAGs.
We say that for a~set of indices $\I$, $\dagg_1$ matches $\dagg_2$ on $\I$ if
$S_i'= S_i$ for each $i \in \I$.
An example illustrating the notions is given in \cref{fig:dags}.
We say that an $\I$-narrow RRDAG $\gendag$ is \emph{accepting} wrt~$\varphi$ if
it matches the (standard) run DAG $\dagw$ on $\I$ and there is some $k \geq 0$ such that
there is a~run $\rho = q_k q_{k+1}\ldots$ with $\rho \models
\varphi$ and for all $i \geq k$ it holds that $q_i \in S_i$ and $q_{i+1} \in
\Delta(q_i, w_i)$.
Intuitively, $\I$~will be positions where we re-sample~$P$, so the level will
be the same as in the full run DAG.

\begin{lemma}\label{lem:fin-accept}
  For every word $\word \in \Sigma^\omega$, it holds that $\word\in \langof{\aut}$ iff the 
  narrow RRDAG $\gendag$ is accepting wrt~$\accfinof{\taccgof{c}} \land \varphi$.
\end{lemma}
\begin{proof}
  ($\Rightarrow$)
  Let $w \in \langof{\aut}$ be a~word. Then there exists a~run $\rho$ of $\aut$
  on $w$ such that $\rho \models \accfinof{\taccgof{c}} \land \varphi$.
  In the narrow RRDAG $\gendag = (S_0, S_1, \ldots)$ it holds that $S_0 = I$.
  Since for all $i$ it holds that $\Delta(S_i, w_i) = S_{i+1}$, $\gendag$ contains $\rho$ and is
  therefore accepting wrt $\accfinof{\taccgof{c}} \land \varphi$. 
  
  ($\Leftarrow$)
  If $\dagg_\word^\delta = (S_0, S_1 \ldots)$ is accepting wrt $\accfinof{\taccgof{c}} \wedge \varphi$, then there is some $k \geq 0$ such that there is a~run $\rho = q_kq_{k+1} \ldots$ with $\rho \models \accfinof{\taccgof{c}} \land \varphi$ and for all $i \geq k$ it holds that $q_i \in S_i$ and $q_{i+1} \in \Delta(q_i, w_i)$. Since $S_0 = I$, then there exists a~run $\rho' = q_0\ldots q_kq_{k+1}\ldots$ which is also a~run of $\aut$ and therefore $w \in \langof{\aut}$. 
  \qed
\end{proof}

\begin{lemma}
  \label{lem:two_accepting_run_dags}
  Let $\word$ be a word and $\gendag$ be the narrow RRDAG.
  Furthermore, let $\I_\daggh$ and $\I_\daggk$ be two infinite sets of indices and $\gendagh$ 
  and $\gendagk$ be two $\I_\daggh$-narrow ($\I_\daggk$-narrow) RRDAGs
  that match $\gendag$ on $\I_\daggh$ ($\I_\daggk$).
  Then, $\gendagh$ is accepting wrt $\varphi$ iff $\gendagk$ is accepting wrt $\varphi$. 
\end{lemma}
\begin{proof}
  ($\Rightarrow$)
  If $\daggh_\word^{\Delta} = (H_0, H_1, \ldots)$ is accepting wrt $\varphi$, then there is some $k \geq 0$ such that there is a~run $\rho = q_kq_{k+1} \ldots$ with $\rho \models \varphi$ and for all $i \geq k$ it holds that $q_i \in H_i$ and $q_{i+1} \in \Delta(q_i, w_i)$.
  Since $\daggh_\word^\Delta$ matches $\dagg_\word^\Delta = (S_0, S_1, \ldots)$ on $\I_\daggh$, there is some $l \geq k$ such that $l \in \I_\daggh$ and $H_l = S_l$. Since $\daggk_\word^\Delta = (K_0, K_1, \ldots)$ matches $\dagg_\word^\Delta$ on $\I_\daggk$, there is some $m \geq l$ such that $m \in \I_\daggk$ and $K_m = S_m$. 
  That means that there is an accepting~run $\rho' = q_m q_{m+1} \ldots$ and $\daggk_\word^\Delta$ is accepting wrt $\varphi$. 

  ($\Leftarrow$) A~similar reasoning as in ($\Rightarrow$) can be used. 
  \qed
\end{proof}

\begin{lemma}\label{lem:dag-accept}
  Let $\word$ be a word. 
  The narrow RRDAG $\gendag$ is \emph{not} accepting wrt $\varphi \wedge \accfinof{\taccgof{c}}$ iff 
    there is an infinite set of indices $\I$ such that the $\I$-narrow RRDAG $\daggh_\word^{\Delta}$ 
    that matches $\dagg_\word^\Delta$ on $\I$ is \emph{not} accepting wrt $\varphi \wedge \accfinof{\taccgof{c}}$.  
\end{lemma}
\begin{proof}
  ($\Rightarrow$)
  $\dagg_\word^\Delta = (S_0, S_1, \ldots)$ is not accepting wrt $\varphi \wedge \accfinof{\taccgof{c}}$, hence there is no  $k \geq 0$ such that
  there is a~run $\rho = q_k q_{k+1}\ldots$ with $\rho \models
  \varphi \wedge \accfinof{\taccgof{c}}$ and for all $i \geq k$ it holds that $q_i \in S_i$ and $q_{i+1} \in
  \Delta(q_i, w_i)$. Let $\daggh_\word^\Delta = (H_0, H_1, \ldots)$. 
  For an infinite set of indices $\I$, it holds that $S_i = H_i$ for every $i \in \I$. There is therefore no accepting run in $\daggh_\word^\Delta$.  

  ($\Leftarrow$)
  Since all runs in the narrow RRDAG $\dagg_\word^\Delta$ are also present in $\daggh_\word^\Delta$ and $daggh_\word^\Delta$ is not accepting, $\dagg_\word^\Delta$ is also not accepting. 
  \qed
\end{proof}

\begin{lemma}\label{lem:run-accept}
  Let $\instanDeltavarphi$ be a~correct subprocedure, $\word$ be a~word, and
  $\gendag = (S_0, S_1, \dots)$ be the narrow RRDAG.
  Furthermore, let 
  $\rho = (S_0, P_0, \mst_0)(S_1, P_1, \mst_1)\ldots$ be a run of $\instancompl(\instanDeltavarphi, \aut)$ over $\word$ and $\I = \{ i \in\omega \mid S_i = P_i \}$.
  Then $\dagg_\rho^{\Delta} = (P_0, P_1, \dots)$ is an $\I$-narrow RRDAG that matches $\gendag$ on $\I$.
\end{lemma}
\begin{proof}
    In order to show that $\dagg_\rho^{\Delta}$ is an $\I$-narrow run DAG that matches $\dagg_\word^\Delta$ on $\I$, we need to show that for indices $i+1 \in \I$ it holds that $\Delta(P_i, w_i) \subseteq P_{i+1}$, for indices $i+1 \not\in \I$ it holds that $\Delta(P_i, w_i) = P_{i+1}$ and for $i \in \I$ it holds that $S_i = P_i$. 
    If $i+1 \in \I$, then $P_{i+1} = S_{i+1}$. Since $P_i \subseteq S_i$, it holds that $\Delta(P_i, w_i) \subseteq P_{i+1}$. 
    If $i+1 \not\in \I$, then $S_{i+1} \neq P_{i+1}$, therefore $P_{i+1} = \Delta(P, a)$ (follows directly from the procedure). 
    $\dagg_\rho^\Delta$ matches $\dagg_w^\delta$ on $\I$ because for every $i \in \I$ it trivially holds that $S_i = P_i$. 
  \qed
\end{proof}

\thmModCorrect*
  

\begin{proof}
  Let $w \in \langof{\instancompl(\instan_\Delta, \aut)}$. Then there is an accepting run $\rho = (S_0, P_0, \mst_0)(S_1, P_1, \mst_1) \ldots $ of $\instancompl(\instan_\Delta, \aut)$ s.t. $\breakempty(\mst_i)$
  holds for infinitely many $i$s. From \cref{lem:run-accept} we have that the narrow RRDAG $\dagg_\rho^\Delta = (P_0, P_1, \dots)$ matches $\dagg_w^\Delta$ on indices $\I$ and from the construction of 
  $\instancompl(\instan_\Delta, \aut)$ we have that $\I$ is infinite (there are infinite many $\breakempty(\mst_i)$ inducing resampling of $P$-part of the macrostate). From correctness of 
  $\instan_\Delta$ we have that $\dagg_\rho^\Delta$ is not accepting wrt $\varphi$. Further from \cref{lem:dag-accept} we have that $\dagg_w^\Delta$ is not accepting wrt $\varphi\wedge \accfinof{\taccgof{c}}$.
  Finally, from \cref{lem:fin-accept} we have that $w\notin \langof{\aut}$.

  Now let $w \in \Sigma^\omega \setminus \langof{\aut}$. Then, $\dagg_\word^\Delta = (S_1, S_2, \dots)$ is not accepting wrt $\varphi\wedge \accfinof{\taccgof{c}}$. From \cref{lem:dag-accept} we have that there 
  is an infinite set of indices $\I$ and $\I$-narrow RRDAG $\daggh_\word^\Delta$ matching $\dagg_\word^\Delta$ and moreover $\daggh_\word^\Delta$ is not accepting wrt $\varphi \wedge \accfinof{\taccgof{c}}$.
  From the correctness of $\instan_\Delta$ we further have that $\daggh_\word^\Delta = (P_1, P_2, \dots)$ is accepting in $\instan_\Delta$ meaning there is a run $\rho = (\mst_{0}, \mst_1, \dots)$ in $\instan_\Delta$
  s.t. $\breakempty(\mst_i)$ for infinite many $i$s. We inductively construct run $R$ of $\instancompl(\instan_\Delta, \aut)$, starting from $R = \epsilon$ as follows: 
  Let $i$ be the first position where $\breakempty(\mst_i)$ holds. We set $R := R . (S_0, P_0, \mst_0) \dots (S_i, P_i, \mst_i)$. Then, $\dagg_{\word[i:]}^\delta = (S_i, S_{i+1}, \dots)$ is not accepting wrt $\varphi\wedge \accfinof{\taccgof{c}}$ and we can again apply \cref{lem:dag-accept} to get an RRDAG $\dagg_{\word[i:]}^\Delta$, which is not accepting wrt $\varphi$. We can hence repeatedly construct another run of $\instan_\Delta$ over $\dagg_{\word[i:]}^\delta$ and mimic the construction of another part of $R$ in the same way as depicted. From the construction, $R$ is accepting meaning that $w\in\langof{\instancompl(\instan_\Delta, \aut)}$.
  \qed 
\end{proof}

\vspace{-0.0mm}
\section{Proofs for \cref{sec:instantiations}}\label{sec:label}
\vspace{-0.0mm}

\lemInstantTrueCorrect*

\begin{proof}
  In order to show that the subprocedure $\instantrue$ is correct for $\mytrue$, we need to show that for each word $w$ and every RRDAG $\dagg_\word^\Delta$ it holds that $\dagg_\word^\Delta$ is not accepting wrt $\mytrue$ iff there is an accepting $\Fin$-run of $\instantrue$ over $\dagg_\word^\Delta$. 
  We first prove this statement from left to right. Assume that $\dagg_\word^\Delta = (S_1,S_2,\ldots)$ is not accepting wrt $\mytrue$, i.e., there is no run $\rho = q_kq_{k+1}\ldots$ for $k \geq 0$ such that for every $i \geq k$ it holds that $q_i \in S_i$ and $q_{i+1} \in \Delta(q_i, w_i)$ and $\rho \models \mytrue$.
  Then, a~$\Fin$-run R of $\instantrue$ over $\dagg_\word^\Delta$ is a~sequence $(M_0,M_1,\ldots)$ such that $M_0 = I$ and for all $i$ it holds that if $M_i = \emptyset$, then $M_{i+1} = S_{i+1}$ and if $M_i \neq \emptyset$, then $M_{i+1} = \Delta(M_i, w_i)$. 
  Since all runs on the input automaton contain infinitely many $\taccgof{c}$-transitions, for all $k > 0$, there is some $j > k$ such that $M_j = \emptyset$ and $\breakempty(M_j)$ is true. The $\Fin$-run R is therefore accepting. 

  Now we prove the statement from right to left. Consider an accepting $\Fin$-run $R = (M_0, M_1, \ldots)$ of $\instantrue$ over $\dagg_\word^\Delta$. It holds that $\breakempty(M_i)$ is true for infinitely many $i$'s. That means that all sampled  runs eventually end when using $\Delta$ as a~transition function, because all runs contain infinitely many accepting states. The word $w$ is therefore not accepted by the input automaton and $\dagg_\word^\Delta$ is not accepting. 
  \qed
\end{proof}

\lemInstantInfCorrect*

\begin{proof}[Sketch]
  In order to show that the subprocedure $\instinf$ is correct, we need to show that for each word $w$ and every RRDAG $\dagg_\word^\Delta$ it holds that $\dagg_\word^\Delta$ is not accepting wrt $\Infof{\tacc{1}}$ iff there is an accepting $\Fin$-run of $\instinf$ over $\dagg_\word^\Delta$. 
  We begin with the proof of the statement from left to right. 
  Assume that $\dagg_\word^\Delta$ is not accepting wrt $\Infof{\tacc{1}}$. 
  There is either no run of $\dagg_\word^\Delta$ on $\word$ at all or all runs do not satisfy the formula. If there is no run of $\dagg_\word^\Delta$ on $\word$, then there is a~sequence $(M_0, M_1, \ldots)$ where $M_0 = \inits$ and $M_{j+1} = \Delta(M_j, a)$ for all $j\geq0$ such that there is some $i \geq 0$ such that $M_l = \emptyset$ for all $l \geq i$. The predicate $\breakempty(M_l)$ is true for all $l \geq i$, so it holds infinitely often, and there therefore exists an accepting run of $\instinf$ over $\dagg_\word^\Delta$. 
  Now assume that there is a~run of $\dagg_w^\Delta$ on $w$. Then, no matter from which point there are no transitions from $\taccgof{c}$, the condition $\accinfof{\tacc 1}$ does not hold for the particular run. With every transition $(f,i) \rightarrow (f',O',i')$ we sample all currently reachable states and then check that all runs from these states contain transitions from $\tacc 1$ only finitely often by modified Schewe's rank-based algorithm. The $O$-component is emptied infinitely often and there is therefore an accepting run of $\instinf$ over $\dagg_\word^\Delta$. 
  
  Now we prove the equivalence in the opposite direction. Assume that there is an accepting run of $\instinf$ over $\dagg_\word^\Delta$. There is therefore a~run where the $\breakempty$ predicate is true infinitely many times. 
  The first possible option is that $\breakempty(P)$ is true infinitely many times. That can happen only if there is no run on $\word$ and $\dagg_\word^\Delta$ is finite. If there is no such run, the formula is not satisfied and $\dagg_\word^\Delta$ is not accepting. 
  The second option is that $\breakempty((f,O,i))$ is true infinitely many times. That means that the formula $\accinfof{\tacc 1}$ does not hold for any run, no matter when the run stops containing transitions from $\taccgof{c}$. The formula is therefore not satisfied in any run and $\dagg_\word^\Delta$ is not accepting. 
  \qed
\end{proof}

\vspace{-0.0mm}
\section{Example of Complementation of Rabin Automata}\label{acc:rabin-example}
\vspace{-0.0mm}

We give an example of the complementation of Rabin automata using subprocedure
$\instinf$ in \cref{fig:rabin-ex}.

\begin{figure}[t]
  \begin{subfigure}[b]{0.3\linewidth}
    \begin{center}
      \scalebox{1.0}{
        \begin{tikzpicture}[>=stealth',shorten >=0pt,scale=0.8,transform shape,initial text={},node distance=2cm]
    \tikzstyle{every state}=[inner sep=3pt,minimum size=5pt]
    \tikzstyle{empty}=[]
    \tikzstyle{initstate}=[fill=yellow!30]
    \tikzstyle{uberstate}=[
      rounded corners,draw,anchor=base,
      rectangle split,rectangle split horizontal,rectangle split parts=3,
      rectangle split part align=base,
      rectangle split part fill={black!20, blue!30, green!30}]
    \newcommand{\ustate}[3]{$#1$\nodepart{two}$#2$\nodepart{three}$#3$}
  
    \node[state,initial] (p) at (0,0) {$p$};
    \node[state,right of=p] (q) {$q$};

    \path[->]
        (p) edge[loop above] node {$a$} (p)
        (p) edge node[above] {$a$} (q)
        (q) edge[loop below] pic{l=$b$} pic[anchor=center] {acc=0} (q)
        (q) edge[loop above] pic{l=$c$} pic[anchor=center] {acc=1} (q);
  \end{tikzpicture}

  
      }
    \caption{Example of a Rabin automaton with the acceptance condition $\accfinof{\tacc 0}\wedge \accinfof{\tacc 1}$.}
    \end{center}
    \label{label}
  \end{subfigure}
  ~
  \begin{subfigure}[b]{0.66\linewidth}
  \begin{center}
    \scalebox{0.6}{
      \begin{tikzpicture}[>=stealth',shorten >=0pt,scale=0.8,transform shape,initial text={},node distance=4cm]
    \tikzstyle{every state}=[inner sep=3pt,minimum size=5pt]
    \tikzstyle{empty}=[]
    \tikzstyle{initstate}=[fill=yellow!30]
    \tikzstyle{uberstate}=[
      rounded corners,draw,anchor=base,
      rectangle split,rectangle split horizontal,rectangle split parts=3,
      rectangle split part align=base,
      rectangle split part fill={black!20, blue!30, green!30}]
    \newcommand{\ustate}[3]{$#1$\nodepart{two}$#2$\nodepart{three}$#3$}
  
    \node[uberstate,initial] (p) at (0,0) {\ustate{\{p\}}{\{p\}}{\{p\}}};
    \node[uberstate,right of=p] (emp) {\ustate{\emptyset}{\emptyset}{\emptyset}};
    \node[uberstate,below of=p,node distance=1.5cm] (pq) {\ustate{\{p,q\}}{\{p,q\}}{\{p,q\}}};
    \node[uberstate,below of=emp,node distance=1.5cm] (qemp) {\ustate{\{q\}}{\{q\}}{\emptyset}};
    \node[uberstate,right of=qemp] (q) {\ustate{\{q\}}{\{q\}}{\{q\}}};

    \node[uberstate,below of=pq,node distance=1.5cm,xshift=20mm] (m1) {\ustate{\{p,q\}}{\{p,q\}}{\{p{:}1,q{:}1\}, 0}};
    \node[uberstate,right of=m1,node distance=6cm] (m2) {\ustate{\{p,q\}}{\{p,q\}}{\{p{:}1,q{:}1\}, \emptyset, 0}};

    \node[uberstate,below of=m1,node distance=1.5cm] (m3) {\ustate{\{p,q\}}{\{p,q\}}{\{p{:}1,q{:}0\}, 0}};
    \node[uberstate,below of=m2,node distance=1.5cm] (m4) {\ustate{\{p,q\}}{\{p,q\}}{\{p{:}1,q{:}0\}, \{q\}, 0}};
    \node[uberstate,right of=m4,node distance=6cm] (m5) {\ustate{\{p,q\}}{\{p,q\}}{\{p{:}1,q{:}0\}, \emptyset, 0}};

    \node[uberstate,below of=m3,node distance=1.5cm] (m6) {\ustate{\{p,q\}}{\{p,q\}}{\{p{:}3,q{:}1\}, 0}};
    \node[uberstate,below of=m4,node distance=1.5cm] (m7) {\ustate{\{p,q\}}{\{p,q\}}{\{p{:}3,q{:}1\}, \emptyset, 0}};
    \node[uberstate,below of=m6,node distance=1.5cm] (m8) {\ustate{\{p,q\}}{\{p,q\}}{\{p{:}3,q{:}1\}, 2}};
    \node[uberstate,below of=m7,node distance=1.5cm] (m9) {\ustate{\{p,q\}}{\{p,q\}}{\{p{:}3,q{:}1\}, \emptyset, 2}};
    
    \path[->]
        (p) edge node[left] {$a$} (pq)
        (p) edge pic {acc=0} pic[auto] {l={$b,c$}} (emp)
        (emp) edge[loop right] pic{l={$a,b,c$}} pic[anchor=center] {acc=0} (emp)
        (pq.12) edge[out=120,in=60,loop,distance=6mm] node[auto] {$a$} (pq.11)
        (pq) edge node[above] {$b$} (qemp)
        (qemp) edge[out=10,in=170] node[above] {$c$} (q)
        (q) edge[in=-10,out=190] pic {acc=0} pic[auto] {l=$b$} (qemp)
        (qemp.30) edge[out=120,in=60,loop,distance=6mm] pic {acc=0} pic[auto] {l=$b$} (qemp.29)
        (qemp) edge pic {acc=0} pic[auto] {l=$a$} (emp)
        (q) edge pic {acc=0} pic[above right] {l=$a$} (emp)
        (q.30) edge[out=120,in=60,loop,distance=6mm] node[auto] {$c$} (q.29)
        (pq) edge node {$a$} (m1)
        (m1.12) edge[out=120,in=60,loop,distance=6mm] node[auto] {$a$} (m1.11)
        (m3.12) edge[out=120,in=60,loop,distance=6mm] node[auto] {$a$} (m3.11)
        (m6.12) edge[out=120,in=60,loop,distance=6mm] node[auto] {$a$} (m6.11)
        (m8.12) edge[out=120,in=60,loop,distance=6mm] node[auto] {$a$} (m8.11)
        (m8) edge[out=10,in=170] pic {acc=0} pic[auto] {l=$a$} (m9)
        (m9) edge[in=-10,out=190] node[above] {$a$} (m8)
        (pq) edge[out=220,in=170] node[right] {$a$} (m3)
        (pq) edge[out=210,in=170] node[right] {$a$} (m6)
        (pq) edge[out=200,in=170] node[right] {$a$} (m8)
        (m1) edge[bend left] node[left] {$a$} (m3)
        (m3) edge[bend left] node[left] {$a$} (m1)
        (m6) edge[out=10,in=170] pic {acc=0} pic[above] {l=$a$} (m7)
        (m7) edge[in=-10,out=190] node[above] {$a$} (m6)
        (m1) edge[out=10,in=170] pic {acc=0} pic[auto] {l=$a$} (m2)
        (m2) edge[in=-10,out=190] node[above] {$a$} (m1)
        (m1) edge node[above] {$a$} (m4)
        (m4) edge pic[pos=0.3] {acc=0} pic[pos=0.3,left] {l=$a$} (m2)
        (m3) edge node[above] {$a$} (m4)
        (m4) edge pic {acc=0} pic[below] {l=$a$} (m5)
        (m5) edge node[pos=0.4,above] {$a$} (m1)
        (m5) edge[out=190,in=-10] node[above] {$a$} (m3);
  \end{tikzpicture}
    }
  \caption{The resulting complementary automaton with the acceptance condition $\accinfof{\tacc 0}$. The macrostates are of the form $S$ (grey), $P$ (blue), $\mst$ (green).}
  \end{center}
  \label{label}
  \end{subfigure}
  \caption{Example of the $\instancompl$ instantiated with $\instinf$ for complementation of automata 
  with the acceptance condition containing a single Rabin pair.
   }
  \label{fig:rabin-ex}
\end{figure}
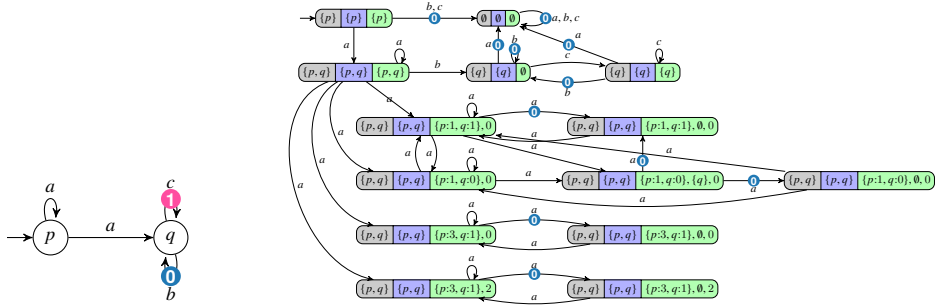

\vspace{-0.0mm}
\section{Generalized Rabin Automata}\label{app:grabin}
\vspace{-0.0mm}

In this section, we give a~subprocedure $\instbinf$ for $\bigwedge_{j = 1}^n\accinfof{\taccj}$ of the modular procedure with an 
algorithm allowing to complement generalized Rabin automata automata with a single generalized Rabin pair.
The macrostates are given as $\minbinf = 2^\states \cup (\cT\times \{ 0, 2, \dots, 2n - 2 \} \times \levelmodels ) \cup 
(\cT\times 2^\states \times \{ 0, 2, \dots, 2n - 2 \} \times \levelmodels)$ where $|\states| = n$ and $\minbinf_0 = \{ I \}$.
The components are then defined as follows:

\begin{minipage}{\textwidth}\scriptsize
  \begin{multicols}{2}
\begin{itemize}
  \item $(f', O', i', \mu') \in \succactbinf_\Delta(P, a, (f, O, i, \mu))$ iff
  \begin{itemize}
    \item $f \finsucca{\Delta,\mu,\mu'} f'$ and $\rankof f = \rankof{f'}$,
    \item $\domof{f'} = P$,
    \item $O \neq \emptyset$ 
    \item $i' = i$, 
    \item $O' = \delta(O, a) \cap f'^{-1}(i)$
  \end{itemize}

  \item $(f', i', \mu') \in \succactbinf_\Delta(P, a, (f, O, i, \mu))$ iff
  \begin{itemize}
    \item $f \finsucca{\Delta,\mu,\mu'} f'$ and $\rankof f = \rankof{f'}$,
    \item $O = \emptyset$ 
    \item $i' = (i+2) \mod (\rankof{f'} + 1)$
  \end{itemize}

  \item $P' \in \succactbinf_\Delta(P, a, P)$ iff
  \begin{itemize}
    \item $P' = P$
  \end{itemize}

  \item $P' \in \succtrackbinf_\Delta(P, a, P)$ iff
  \begin{itemize}
    \item $P' = P$
  \end{itemize}

  \item $(f', i',\mu') \in \succtrackbinf_\Delta(P, a)$ iff
  \begin{itemize}
    \item $f'$ is $(P,\mu')$-tight
    \item $\mu'\in \levelmodels$,
    \item $i'= 0$
  \end{itemize}

  \item $\{(f', O', i', \mu'), (f', i', \mu')\} \subseteq \succtrackbinf_\Delta(P, a, (f, i, \mu))$ iff
  \begin{itemize}
    \item $f \finsucca{\Delta,\mu,\mu'} f'$ and $\rankof f = \rankof{f'}$,
    \item $O' = f'^{-1}(i)$,
    \item $\mu' \in \levelmodels$,
    \item $i'= i$
  \end{itemize}
  \item $\breakemptybinf((f,O,i,\mu)) \Longleftrightarrow O = \emptyset$
  \item $\breakemptybinf(P) \Longleftrightarrow P = \emptyset$
  \item $\neg \breakemptyinf((f,i,\mu))$
\end{itemize}
\end{multicols}
\end{minipage}
\vspace{2mm}

\begin{lemma}
  The subprocedure $\instbinf = (\mstbinf, \mstbinf_0, \succactbinf_\Delta, \succtrackbinf_\Delta, \breakemptybinf)$ for $\bigwedge_{j = 1}^n \accinfof{\taccj}$ is correct.
\end{lemma}
\begin{proof}[Sketch]
  In order to show that the subprocedure $\instbinf$ is correct, we need to show that for each word $\word$ and every RRDAG $\dagg_\word^\delta$ over $w$ it holds that $\dagg_\word^\delta$ is not accepting wrt $\varphi$ iff there exists an accepting run of $\instbinf$ over $\dagg_\word^\Delta$. 
  We begin with the proof of the equivalence from left to right. 
  Assume that $\dagg_\word^\Delta$ is not accepting. 
  There is either no run on $\word$ at all or all runs do not satisfy the formula. If there is no run of $\dagg_\word^\Delta$ on $\word$, there is a~sequence $(M_0, M_1, \ldots)$ where $M_0 = \inits$ and $M_{j+1} = \Delta(M_j, a)$ for all $j\geq0$ such that there is some $i \geq 0$ such that $M_l = \emptyset$ for all $l \geq i$. The predicate $\breakempty(M_l)$ is true for all $l \geq i$ and there therefore exists an accepting run of $\instan_\delta$ over $\dagg_\word^\Delta$. 
  Now assume that there is a~run of $\dagg_w^\Delta$ on $w$. Then, no matter from which point there are no transitions from $\taccgof{c}$, the condition $\bigwedge_{j = 1}^n \accinfof{\taccj}$ does not hold for the particular run. With every transition $(f,i, \mu) \rightarrow (f',O',i', \mu')$ we sample all currently reachable states and then check that the condition $\bigwedge_{j = 1}^n \accinfof{\taccj}$ does not hold for any run from these states by modified Schewe's rank-based algorithm for GBAs described in \cref{subsec:infela}. The $O$-component is emptied infinitely often and there is therefore an accepting run of $\instbinf$ over $\dagg_\word^\Delta$. 
  
  Now we prove the equivalence in the opposite direction. Assume that there is an accepting run of $\instbinf$ over $\dagg_\word^\Delta$. There is therefore a~run where the $\breakempty$ predicate is true infinitely many times. 
  The first possible option is that \\ $\breakempty(P)$ is true infinitely many times. That can happen only if there is no run on $\word$ and $\dagg_\word^\Delta$ is finite. If there is no such run, the formula is not satisfied and $\dagg_\word^\Delta$ is not accepting. 
  The second option is that $\breakempty((f,O,i))$ is true infinitely many times. That means that the formula $\bigwedge_{j = 1}^n \accinfof{\taccj}$ does not hold for any run, no matter when the run stops containing transitions from $\taccgof{c}$. The formula is therefore not satisfied in any run and $\dagg_\word^\Delta$ is not accepting. 
  \qed
\end{proof}

\begin{lemma}
  \label{lem:generalized_rabin}
  Let $\aut$ be a generalized Rabin automaton with one generalized Rabin pair with $\ell$ $\accinf$s. Then, the 
  complemented GBA has $\bigO(\ell^{n}\tight(n+1))$ states. 
\end{lemma}
\begin{proof}
  It suffices to count the number of macrostates of the form $(S, P, f, O, i, \mu)$. Consider a~macrostate $(S, P, f, O, i, \mu)$. 
We uniquely encode each macrostate as $(h, i)$ where $h \colon \states \rightarrow \{ -3, -2, -1, \ldots, 2n-1 \} \times \modelsof \acccond$ for $n = |\states|$ is defined as follows:
 \begin{equation}
  h(q) = \begin{cases}
    (-1, \mu) & \text{if } q \in O, \\
    (-2, \mu) & \text{if } q \in \states\setminus S, \\
    (-3, \mu) &\text{if } q \in S \setminus P, \text{and} \\ 
    (f(q), \mu) & \text{otherwise}.
  \end{cases}
\end{equation}
We compute the number of encodings $h$ for a~fixed $i$. We divide all encodings into groups according to the set $img(h)_0 \cap \{ -3, -2, -1 \}$ where $img(h)_0$ denotes the set of first elements of the pairs in $img(h)$. We show that for each of the (at most 8) groups we can obtain the bound $\bigOof{l^n \cdot \mathit{tight}(n)}$. 
Each of the group is denoted by $g_M$ with $M \subseteq \{-2, -1\}$,
  i.e., $g_M = \{ h\colon \states \to \{-3, -2, \ldots, 2n-1\} \times \modelsof \acccond \mid M = \img(h)_0 \cap \{-3, -2, -1\}\}$.
For $h(q) = (m, \mu)$, we denote $m$ by $h(q)(0)$ and $\mu$ by $h(q)(1)$. 

  \begin{enumerate}
    \item[$g_\emptyset$:] from the definition, $f$~is $\mu$-tight. The level model $\mu$ is of the form $\mu \colon \states \rightarrow \modelsof \acccond$, so there are $l$ possible assignments for every state from $\states$. The number of level models is therefore $l^n$ and $|g_\emptyset| = \bigO(l^n \cdot \tight(n))$. 
    
    \item[$g_{\{-1\}}$:] since there is at least one state $q$ with $h(q)(0) = -1$,
      this means that $q \in O$ so~$q$ has an even rank.
      As a~consequence, at least one of the positive odd ranks of~$h$ will not
      be taken, so we can infer that $h\colon \states \to \{-1, \ldots, 2n-3\} \times \modelsof \acccond$.
      We can therefore uniquely map~$h$ to a~mapping~$h'$ by incrementing all
      ranks of~$h$ by two, so $h'\colon \states \to \{0, \ldots, 2n-1\} \times \modelsof \acccond$.
      But then $h' \in \cT(n)$ and the number of all level models is $l^n$, so $|g_{\{-1\}}| \in \bigO(l^n \cdot \tight(n))$.

    \item[$g_{\{-2,-1\}}$:] via the same reasoning as for $g_{\{-1\}}$ we get
      that $|g_{\{-2,-1\}}| \in \bigO(l^n \cdot \tight(n))$.

    \item[$g_{\{-2\}}$:] the reasoning is similar to the one for  $g_{\{-1\}}$,
      with the exception that now, we know that there is a~state~$q \in Q
      \setminus S$, which is, according to the definition of a~ranking, assigned
      the rank~$0$.
      This means that one positive odd rank of~$h$ is, again, not taken, so we
      increment all non-negative ranks of~$h$ by two and map all states in $Q
      \setminus S$ to~$1$, obtaining a~tight ranking $h' \in \cT(n)$.
      The number of level models is $l^n$, 
      therefore, $|g_{\{-2\}}| \in \bigO(l^n \cdot \tight(n))$.

      \item[$g_{\{-3\}}$:] the reasoning is similar to the one for $g_{\{-1\}}$, with the exception that now we know that there is some state $q \in S \setminus P$ such that its rank is, according to the definition, $0$. Therefore, we increment all non-negative ranks of $h$ by two and map the states in $P \setminus \states$ to 1, obtaining a~tight ranking $h' \in \cT(n)$. For $l^n$ possible level models, it holds that $|g_{\{-3\}}| \in \bigOof{l^n \cdot \tight(n)}$. 
      
      \item[$g_{\{-3,-2\}}$, $g_{\{-3, -1\}}$:]
      similarly as for $g_{\{-2\}}$, we increment all non-negative ranks of~$h$
      by two and set $h'(q)(0) = 0$ if $h(q)(0) = -3$ and $h'(q)(0) = 1$ if $h(q)(0) = -2$
      (resp.\ if $h(q)(0) = -1$).
      Then $h' \in \cT(n)$ and so for $l^n$ level models it holds that $|g_{\{-3,-2\}}| \in \bigO(l^n \cdot \tight(n))$ and
      $|g_{\{-3, -1\}}| \in \bigO(l^n \cdot \tight(n))$.

    \item[$g_{\{-3, -2, -1\}}$:] in this case, we know that there is at least
      one state $q_1 \in O$ and at least one state $q_2 \in \states \setminus S$.
      Therefore, there will be at least two odd positions not taken in~$h$, so
      we can infer that $h\colon \{-3, \ldots 2n-5\}$.
      We create $h'$ by incrementing all ranks in~$h$ by \emph{four}; in this
      way, we obtain a~tight ranking $h'\colon \states \to \{0, \ldots, 2n-1\}$, so
      for $l^n$ level models it holds that $|g_{\{-3, -2, -1\}}| \in \bigO(l^n \cdot \tight(n))$.
  \end{enumerate}

  Since the size of all groups is bounded by $\bigO(l^n \cdot \tight(n))$, for
  a~fixed~$i$, the total number of these encodings is still
  $\bigO(k^n \cdot \tight(n))$.
  When we sum the encodings for all possible~$i$'s, we obtain that the number is
  bounded by~$\bigO(l^n \cdot \tight(n+1))$, since $\bigOof{n \cdot \tight(n)} = \bigOof{\tight(n+1)}$.
  \qed
\end{proof}

\begin{theorem}
  Let $\aut$ be a generalized Rabin automaton with $k$ generalized Rabin pairs and each 
  pair has at most $\ell$ $\accinf$s. Then, the 
  complemented GBA has $\bigO(\ell^{kn}\tight(n+1)^k)$ states. 
\end{theorem}
\begin{proof}
  Proof follows directly from~\cref{lem:generalized_rabin}. In order to complement a~generalized Rabin automaton with $k$ generalized Rabin pairs with at most $\ell$ $\accinf$s, we construct a~complementary automaton for each generalized Rabin pair and then we make a~product of these automata and obtain a~GBA accepting the complement of the original generalized Rabin automaton. 
  \qed
\end{proof}

\end{document}